\documentclass[letter,12pt]{article}

\usepackage[english]{babel}
\usepackage[utf8]{inputenc}
\usepackage{amsmath}
\usepackage{amssymb}
\usepackage{amsbsy,enumerate}
\usepackage{graphicx}
\usepackage{subcaption}
\usepackage[colorinlistoftodos]{todonotes}
\usepackage[percent]{overpic}
\usepackage{pdflscape}
\usepackage{booktabs}
\usepackage{hyperref}           
\usepackage{setspace}
\usepackage[authoryear]{natbib}
\usepackage{multirow}
\usepackage{tabulary}
\usepackage{float}
\usepackage{overpic}
\usepackage{bm}
\usepackage[font=footnotesize,labelfont=bf]{caption}
 \usepackage{colortbl}
 \usepackage{makecell}
 \usepackage{siunitx}
 \usepackage{threeparttable}
 \usepackage{titlesec}
\hypersetup{
    pdftitle={},                  
    pdfsubject={}, 
    pdfauthor={},              
    pdfkeywords={},       
    plainpages=false, %
    colorlinks,       
    urlcolor=blue,    
    linkcolor=red,    
    citecolor=cyan,  
    bookmarksnumbered
}

\usepackage{dcolumn}
\newcolumntype{d}[1]{D..{#1}} 

\usepackage{xr}
\input{ee.sty}

\newtheorem{theorem}{Theorem}

\newenvironment{proof}[1][Proof]{\begin{trivlist}
\item[\hskip \labelsep {\bfseries #1}]}{\end{trivlist}}

\newcommand{\qed}{\nobreak \ifvmode \relax \else
      \ifdim\lastskip<1.5em \hskip-\lastskip
      \hskip1.5em plus0em minus0.5em \fi \nobreak
      \vrule height0.75em width0.5em depth0.25em\fi}

\usepackage{dcolumn}
\newcolumntype{d}[1]{D{.}{.}{#1}}

\newcommand{\ben}{\begin{displaymath}}
\newcommand{\een}{\end{displaymath}}

\title{Corrected Forecast Combinations} 
\date{\today }
\author{Chu-An Liu\\
\small{\it Academia Sinica, Taiwan}
\\
and\\
Andrey L. Vasnev\thanks{%
Corresponding author contact information: The University of Sydney Business School, \href{mailto:andrey.vasnev@sydney.edu.au}{\texttt{andrey.vasnev@sydney.edu.au}}.  We would like to thank Jan R. Magnus for stimulating conversations spanning multiple years that led to the fruition of this project, Francis X. Diebold for bringing his earlier attempts to our attention, participants of the MyJP-SIM$^2$ conference in Tokyo, the ISF2026 conference in Montreal and the Business Analytics seminar at the University of Sydney, and to the editor and reviewers for their encouragement, deep engagement with the paper and constructive comments. 
The hospitality of the Institute of Economics, Academia Sinica, where the empirical foundations of this project were laid, is very much appreciated by A.L. Vasnev. Chu-An Liu gratefully acknowledges research support from the National Science and Technology Council of Taiwan (NSTC 115-2628-H-001-002). The authors are responsible for any remaining errors.} \\
\small{\it The University of Sydney, New South Wales, Australia}}

\begin{document}

\maketitle
\begin{abstract}
This paper proposes corrected forecast combinations when the original combined forecast errors are serially dependent.
Motivated by the classic \citet{bates1969combination} example, we show that combined forecast errors can be strongly autocorrelated and that a simple correction, which adds a fraction of the previous combined error to the next-period combined forecast, can deliver sizable improvements in forecast accuracy, often exceeding the original gains from combining.
We formalize the approach within the conditional-risk framework of \citet{COW2024}, in which the combined error decomposes into a predictable component (measurable at the forecast origin) and an innovation. We then link this correction to efficient estimation of combination weights under time-series dependence via GLS, allowing joint estimation of weights and an error-covariance structure. Using the U.S. Survey of Professional Forecasters for major macroeconomic indices across various subsamples (including pre/post-2000, GFC, and COVID), we find that a parsimonious correction of the mean forecast with a coefficient around 0.5 is a robust starting point and often yields material improvements in forecast accuracy. {The benefits of correction diminish rapidly as the information becomes stale, largely disappearing at the four-quarter horizon or when the correction relies on an older lagged forecast error.} For optimal-weight forecasts, the correction substantially mitigates the forecast combination puzzle by turning poorly performing out-of-sample optimal-weight combinations into competitive forecasts.


\end{abstract}
\thispagestyle{empty}
\newpage

\doublespace

\section{Introduction}

Consider the motivating example from \citet{bates1969combination}, in which the passenger miles flown in 1953 were predicted by two models, exponential smoothing (ES) and Box-Jenkins' ARIMA model (BJ), as well as by their equally weighted (EW) combination, reproduced in Table~\ref{tbl:BG1969}. The final line in the table reports the mean squared forecasting error (MSFE).
The final column presents our contribution, explained below.

\sisetup{
  table-number-alignment = center,
  table-format = 3.1,        
  detect-all,                
  input-symbols = {−},       
  table-align-text-post = false
}

\begin{table}[htbp]
\centering
\caption{Errors in forecasts (actual less estimated) of passenger miles flown, 1953.}
\label{tbl:BG1969}
\begin{tabular}{
    l
    S[table-format=3.1]
    S[table-format=3.1]
    S[table-format=3.1]
    S[table-format=3.1]
}
\hline
Month &
\multicolumn{1}{c}{\makecell{Brown's exponential\\ smoothing forecast\\ errors}} &
\multicolumn{1}{c}{\makecell{{Box--Jenkins}\\ {adaptive}\\ {forecasting}\\ {errors}}} &
\multicolumn{1}{c}{\makecell{{Combined forecast}\\ ($\tfrac{1}{2}$ Brown +\\ $\tfrac{1}{2}$ Box--Jenkins)\\ {errors}}} &
\multicolumn{1}{c}{\makecell{Corrected\\ combined forecast \\ errors}} \\
\hline
Jan.   & 1   & -3   & -1 & -2.375 \\
Feb.  & 6   & -10  & -2 & -1.5 \\
March & 18  & 24     & 21 & 22 \\
April & 18  & 22     & 20 & 9.5 \\
May   & 3   & -9   & -3 & -13 \\
June  & -17 & -22 & -19.5 & -18 \\
July  & -24 & 10    & -7 & 2.75 \\
Aug.  & -16 & 2     & -7 & -3.5 \\
Sept. & -12 & -11 & -11.5 & -8 \\
Oct.  & -9  & -10 & -9.5 & -3.75 \\
Nov.  & -12 & -12 & -12 & -7.25 \\
Dec.  & -13 & -7  & -10 & -4 \\
\hline
MSFE & 196 & 188 & 150 & 103\\
\hline
\end{tabular}

\begin{tablenotes}
\small
\item {\it Note.} The first four columns are a replication of Table 1 from \citet{bates1969combination} except that the title of the last line is changed from `Variance of errors' to MSFE. Corrected combined forecast errors in the final column are equal to the combined forecast errors$- 0.5\times$combined forecast errors from the previous period, e.g., for Feb. $-1.5=-2-0.5\times (-1)$.
To correct the Jan. observation, we used the December 1952 error, which is 2.75, from \citet{Barnard1963}.
\end{tablenotes}
\end{table}

This motivating example, showing that the average of two forecasts is better than the individual forecasts, started the area of forecast combination, with a multitude of papers appearing every year; see \citet{WANG20231518} for the latest over 50-year review. One simple fact, however, remained overlooked over all these years: the errors of the combined forecast exhibit strong autocorrelation of around 54\%.
This information can be used in different ways, but the simplest option is to correct the combined forecast with the error from the previous period, i.e., if we denote the actual observation $y_t$ at time $t$, the combined forecast $f^{\text{EW}}_{t | t-1}$ based on the information available at time $t-1$, and its error $e^{\text{EW}}_{t | t-1}=y_{t}-f^{\text{EW}}_{t | t-1}$ observable at time $t$, then the combined forecast for the next period $t+1$ can be corrected using
\begin{equation} \label{eq:intro_correction}
   f^{\text{CEW}}_{t+1 | t} = f^{\text{EW}}_{t+1 | t} + 0.5\, e^{\text{EW}}_{t | t-1}
\end{equation}
(where we use the correction of 0.5 for simplicity)
and its error is
$$
  e^{\text{CEW}}_{t+1 | t} = y_{t+1} - f^{\text{CEW}}_{t+1 | t} = e^{\text{EW}}_{t+1 |t} - 0.5\, e^{\text{EW}}_{t | t-1},
$$
which is reported in the final column in Table~\ref{tbl:BG1969}.
What is surprising is that correcting for first-order autocorrelation yields even greater improvement than the original combination. Specifically, the change in MSFE from the best individual forecast to the combination (from 188 to 150) is a 20\% reduction, while the change in MSFE from the combination to the corrected combination (from 150 to 103) is a 31\% reduction.

\begin{figure}[htbp]
\centering
\includegraphics[width=0.7\linewidth]{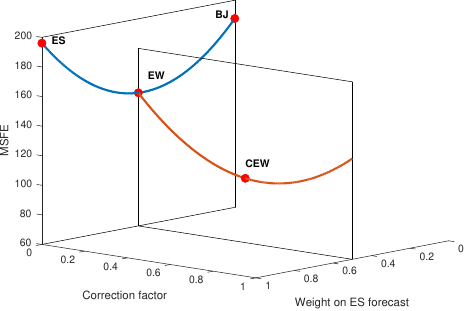}
\caption{Improvements from combination (EW) and correction (CEW) in the motivating example from \citet{bates1969combination}. Points ES and BJ correspond to the original Brown's exponential smoothing and Box-Jenkins forecasts.}
\label{fig:BG}
\end{figure}

Figure~\ref{fig:BG} depicts this improvement. The back plane shows the performance of two original forecasts, ES and BJ, and the equally weighted combination, EW. The quadratic curve in this plane passing through all three points has been studied in \citet{claeskens2016forecast} and holds only when a fixed weight (such as 0.5) is used. Once we correct EW, a new dimension emerges: the plane corresponding to changes in the correction factor. It also contains a new curve that starts at EW and illustrates the improvements achieved by using a fixed correction factor, such as 0.5, which yields point CEW on that curve. The improvement in MSFE from the correction, i.e., from EW to CEW, is visually larger than the improvement from the combination, i.e., from ES or BJ to EW, in line with the numerical results in Table~\ref{tbl:BG1969}.

This motivating example shows that, although the vast forecast combination literature operates in the back plane and has made substantial progress over the years, see \citet{WANG20231518}, a new dimension focusing on the properties of the combined forecast can be highly beneficial.
This dimension was noted by \citet{Diebold1988} almost forty years ago, {extended by \citet{CoulsonRobins1993} and included in the classical textbook of \citet{clements1998forecasting}, Chapter 10.2}. However, only a few authors have engaged with this idea substantively; see, for example, \citet{Gallo1996}, \citet{DEMENEZES2000}, and \citet{InconvenientTruth2023}. Other authors simply mention it, use it to justify dropping the independence assumption, or use it to justify the restriction that the combination weights sum to one. \citet{Diebold1988} also showed that correlated combined forecast errors are common, and \citet{InconvenientTruth2023} recently proved that ``traditional weighted averages of forecasts are almost always auto-inefficient'' and, therefore, the combined forecast errors are predictable.

Furthermore, the optimal forecasting weights introduced by \citet{bates1969combination} exhibit poor empirical performance, leading to the forecast combination puzzle of \citet{SW2004}: ``repeated finding that simple combination forecasts outperform sophisticated adaptive combination methods in empirical applications''. The optimal weight problem can be estimated using ordinary least squares (OLS) in a linear regression, as noted by \citet{granger1984improved}. However, as many textbooks show, autocorrelated errors in regression models make OLS estimators inefficient, potentially leading to undesirable consequences. Generalized least squares (GLS) estimation is more efficient than OLS in the presence of autocorrelated errors, a point recently reiterated by \citet{POOJARI2025100109}. By focusing on the properties of the combined optimal forecast, as advocated by \citet{DEMENEZES2000}, we should be able to shed more light on the forecast combination puzzle.

The correction (\ref{eq:intro_correction}) might look similar to the intercept correction of \citet{clements1996intercept}, where its most commonly used form is described as the last available error being added in at each step ahead forecast over the forecasting period.
{Both methods therefore use past forecast errors as feedback. They differ primarily in motivation and calibration.
While the intercept correction interprets the most recent residual as evidence of a change in the mean, the correction (\ref{eq:intro_correction}) exploits an autoregressive structure whose coefficient can be calibrated using the available history of forecast errors.}
{Although \citet{clements1996intercept} apply their correction to individual models and we focus on combined forecasts, either correction can in principle be applied in both settings.}
{\citet{clements2022} found evidence that many forecasters’ predictions are inefficient and implemented a real-time efficiency correction separately for each forecaster. These corrected forecasts could subsequently be combined, although Clements did not investigate this possibility. The empirical performance of this correct-then-combine approach can be compared with that of our combine-then-correct approach, but the two are not generally equivalent. Moreover, correcting many individual forecasts comes with greater estimation burden and parameter uncertainty than correcting a single mean forecast.}

The paper is organized as follows. Section~\ref{sec:theory} formalizes our idea with the correction of the mean (\ref{sec:mean}) and the optimal forecast (\ref{sec:optimalF}) via the two-step procedure used in the motivating example, before calling on the conditional framework of \citet{COW2024} in Subsection~\ref{sec:GLS} that allows simultaneous estimation of the optimal weights and the correction factor. We also provide a generalization of our approach via GLS in Subsection~\ref{sec:generalization}. {The proofs are in Appendix~\ref{app:proofs}.}
The empirical illustration using the U.S. Survey of Professional Forecasters is presented in Section~\ref{sec:empirical} {with additional empirical results in Appendix~\ref{app:empirical}}. Finally, Section~\ref{sec:conclusion} concludes.

\section{Theory} \label{sec:theory}

Assume that we wish to forecast $y_{T+h}\in\mathbb{R}$ using the vector of $h$-step-ahead forecasts
$\vf_{T+h | T}=(f_{1,T+h | T},f_{2,T+h | T},\dots,f_{n,T+h | T})' \in \mathbb{R}^n$.
Following \citet{GiacominiWhite:2006}, we assume that $y_t$ may follow a complex process marked by measurement issues, structural changes, and nonstationarity induced by distribution changes.
Following \citet{aiolfi2006persistence}, we map all forecasts to the real number line and limit our analysis to linear combinations with weights $
\vw=(w_1,w_2,\dots,w_n)' \in \mathbb{R}^n$
to produce the combined forecast  $f_{c,T+h | T}=\vw'\vf_{T+h | T}$.
We denote the vector of forecasting errors as $\ve_{T+h | T}=y_{T+h}\viota - \vf_{T+h | T}, $
where $\viota$ is a vector of ones, and the error of the combined forecast as $e_{c,T+h | T}=y_{T+h}-f_{c,T+h | T}=\vw'\ve_{T+h | T}.$

We first formalize the correction to the mean forecast used in the motivating example. We then embed it in the conditional framework of \citet{COW2024} and conclude with the generalization.

\subsection{Correction of the mean forecast} \label{sec:mean}

Let us consider the equal weights $\vw^{\text{EW}} = \viota/n$ (where $\viota$ is a $n\times 1$ vector of ones), a generalization of the motivating example with two forecasts in the introduction. The combined forecast $f_{T+h | T}^{\text{EW}}=\viota'\vf_{T+h | T}/n$ and the corresponding forecasting error $e^{\text{EW}}_{T+h | T} = y_{T+h} - f_{T+h | T}^{\text{EW}}$ can be corrected using a correction term $b_T^{\text{EW}}$ available at time $T$ to produce the corrected combined forecast $f_{T+h | T}^{\text{CEW}} = f_{T+h | T}^{\text{EW}} + b_{T}^{\text{EW}}$. In our motivating example, $b_{T}^{\text{EW}} = 0.5 \, e^{\text{EW}}_{T | T-1}$ with $h=1$, a simple function of the previous forecasting error. However, a general correction term $b_{T}^{\text{EW}}(I_T,\vgamma) $ that involves information $I_T$, including additional variables, and parameters $\vgamma$ is also possible and presented in Section~\ref{sec:generalization}.
Here, we show a simple possibility of estimating one parameter.

For $h=1$, when dealing with the simple correction term $b_{T}^{\text{EW}} = \gamma \, e^{\text{EW}}_{T | T-1}$, it seems that the correction factor of 0.5 is a reliable starting point in many applications, but with sufficient data available, one can use historical data to find the optimal correction. For a general forecast horizon $h$, we estimate a separate correction
factor in each period $t$ by solving
\begin{equation}
   {\widehat{\gamma}_{h,t}} = \arg \min_\gamma \sum_{\tau=t_0}^{t-h}\left[y_{\tau+h} - (f^{\text{EW}}_{\tau+h | \tau} + \gamma\, e^{\text{EW}}_{\tau | \tau-h})\right]^2 \label{eq:optCFh}
\end{equation}
and use the historically optimal (in the period $[t_0,t]$) factor ${\widehat{\gamma}_{h,t}}$  to correct forecast  $f^{\text{CEW}}_{t+h | t} = f^{\text{EW}}_{t+h | t} +{\widehat{\gamma}_{h,t}}\, e^{\text{EW}}_{t | t-h}$ at the future moment $t+h$.\footnote{Optimization problem (\ref{eq:optCFh}) has an explicit solution $\widehat{\gamma}_{h,t} = {(\sum_{\tau=t_0}^{t-h} e^{\text{EW}}_{\tau+h | \tau} e^{\text{EW}}_{\tau | \tau-h})} / {(\sum_{\tau=t_0}^{t-h} (e^{\text{EW}}_{\tau | \tau-h})^2)}$, which estimates the $h$-th-order autocorrelation of the previous forecasting errors. An explicit solution for a general correction term $b_{T}^{\text{EW}}(I_T,\vgamma)$ might not exist, but a numerical solution can always be used.}

The optimal correction, in terms of the mean squared error (MSE), is achieved by $b_{T}^{\text{EW}} = \E(e_{T+h | T}^{\text{EW}} | I_T)$ as given by Theorem~\ref{Th1_combined}.
In applications, one needs to specify how to model $b_{T}^{\text{EW}} = \E(e_{T+h | T}^{\text{EW}} | I_T)$ and estimate it, which will affect the MSE. Therefore, keeping the model simple and minimizing additional estimation, as in our motivating example, should be a favorable correction strategy.

\subsection{Correction of the (unconditionally) optimal forecast} \label{sec:optimalF}

We now turn our attention to the optimal weights. Assuming that the loss function $L(\cdot)$ depends only on  $e_{c,T+h | T}$, the classical \emph{unconditionally} optimal weights solve the problem
\begin{equation}
    \min_{\vw} \E[L(e_{c,T+h | T})]. \label{eq:opt_problem_uncond}
\end{equation}
Under mean squared error loss, $L(e)=e^2$, only the first two unconditional moments influence the optimal weights, and the optimization problem can be solved explicitly. We also assume that the forecasts are unbiased, $\E(\ve_{T+h | T})=\vzeros$; thus, we solve the optimization problem~(\ref{eq:opt_problem_uncond}) subject to the restriction that the weights sum up to one, $\vw'\viota=1$, which produces the optimal weights of \citet{bates1969combination}
\begin{equation}
   \vw^{\text{BG}} = \frac{\Sigma^{-1}\viota}{\viota'\Sigma^{-1}\viota}, \label{eq:BGweights}
\end{equation}
where $\Sigma = \E[\ve_{T+h | T}\ve'_{T+h | T}]$, and the optimal BG forecast $f_{T+h | T}^{\text{BG}}=(\vw^{\text{BG}})'\vf_{T+h | T}$.
The existence of the first two moments of the errors, $\ve_{T+h | T}$, is sufficient for the results derived in this section to be valid under our general data assumptions.

If a suitable correction term $b_T^{\text{BG}}$ is available, then the conditional and unconditional inequalities for the reduction in error will hold for the corrected BG forecast $f_{T+h | T}^{\text{CBG}} = f_{T+h | T}^{\text{BG}} + b_T^{\text{BG}} = (\vw^{\text{BG}})'\vf_{T+h | T} + b_T^{\text{BG}}$, as stated in Theorem~\ref{Th1_combined}.
In fact, any other combination $f_{T+h | T}^{\text{ZZ}}$ (not necessarily optimal) or even the original individual forecasts can be corrected in a similar manner to improve performance.

{
\begin{theorem}[Oracle correction] \label{Th1_combined}
Let ${I}_T$ denote the information available at time $T$.
Suppose that $y_{T+h}$ and the forecast $f_{T+h\mid T}^{\text{ZZ}}$
have finite second moments and that $f_{T+h\mid T}^{\text{ZZ}}$ is
${I}_T$-measurable with the forecasting error given by
$
e_{T+h\mid T}^{\text{ZZ}}
=
y_{T+h}-f_{T+h\mid T}^{\text{ZZ}}.
$
Among all square-integrable, ${I}_T$-measurable additive
corrections, the oracle correction term
\[
b_T^{\text{ZZ}}
=
\E\!\left(
e_{T+h\mid T}^{\text{ZZ}}
\mid{I}_T
\right)
\]
minimizes conditional and unconditional MSE.
Consequently, the oracle-corrected forecast
$
f_{T+h\mid T}^{\text{CZZ}}
=
f_{T+h\mid T}^{\text{ZZ}}+b_T^{\text{ZZ}}
=
\E\!\left(y_{T+h}\mid{I}_T\right)
$
has error
$
e_{T+h\mid T}^{\text{CZZ}}
=
y_{T+h}-f_{T+h\mid T}^{\text{CZZ}}
$
which satisfies $\E(e_{T+h\mid T}^{\text{CZZ}}\mid{I}_T )= 0$,
\[
\E\!\left[
  \left(e_{T+h\mid T}^{\text{ZZ}}\right)^2
  \mid{I}_T
\right]
-
\E\!\left[
  \left(e_{T+h\mid T}^{\text{CZZ}}\right)^2
  \mid{I}_T
\right]
\geq 0
\quad\text{a.s.},
\]
and
\[
\E\!\left[
  \left(e_{T+h\mid T}^{\text{ZZ}}\right)^2
\right]
-
\E\!\left[
  \left(e_{T+h\mid T}^{\text{CZZ}}\right)^2
\right]
\geq 0.
\]
The conditional MSE improvement is strict precisely on the event
$
\left\{
\E\!\left(
e_{T+h\mid T}^{\text{ZZ}}\mid I_T
\right)\neq0
\right\},
$
and the unconditional MSE improvement is strict if and only if
$
\Pr\!\left[
\E\!\left(
e_{T+h\mid T}^{\text{ZZ}}\mid I_T
\right)\neq0
\right]>0.
$
\end{theorem}
\begin{proof} See Appendix~\ref{app:proofs}. \end{proof}
}

{
Note that for any ${I}_T$-measurable
forecast, including $\text{ZZ}=\text{EW}$ and
$\text{ZZ}=\text{BG}$, the corrected version is the same
\[
f_{T+h\mid T}^{\text{CZZ}}
=
\E\!\left(y_{T+h}\mid{I}_T\right).
\]
Thus, all such oracle-corrected forecasts coincide, but only if the same information $I_T$ is used.
Theorem~\ref{Th1_combined} is an oracle result. In feasible
applications, modeling and estimation errors may offset some or all
of the oracle improvement, and the feasible correction need not
reduce MSE.}

All previous corrections were two-step procedures, in which the original forecasts were combined in the first step and corrected in the second step.
The correction term, such as $b_{T}^{\text{EW}}$ or $b_{T}^{\text{BG}}$, will generally depend on a vector of unknown parameters $\vgamma\in\mathbb{R}^k$
that need to be estimated or selected. From now on, we explicitly denote this dependence as $b_{T}^{\text{EW}}(\vgamma)$ or $b_{T}^{\text{BG}}(\vgamma)$. In our motivating example in the introduction, this is a scalar first-order autoregressive parameter, and its value is simply set to 0.5. However, if $\vgamma$ is estimated, then the two-step procedure is the simplest option, i.e.,
given the weights, $\vw^{\text{EW}}$, $\vw^{\text{BG}}$ (or any other chosen weights) and their corresponding combinations $f_{T+h | T}^{\text{EW}}$, $f_{T+h | T}^{\text{BG}}$, we can use the corresponding combined forecasting errors to estimate $\vgamma$ and correct the combined forecast $f_{T+h | T}^{\text{CZZ}} = f_{T+h | T}^{\text{ZZ}} + b_{T}^{\text{ZZ}} (\widehat{\vgamma}^{\text{ZZ}} )$, where ZZ stands for EW, BG or any other combination that one decides to use. This procedure may be effective in many situations due to its simplicity. For example, \citet{Gallo1996} used a similar AR(p) correction procedure for individual forecasts. 

\subsection{Conditional Framework of \citet{COW2024}} \label{sec:GLS}

One expects benefits from jointly estimating $\vw$ and $\vgamma$, at least in theory, so
we now turn our attention to a one-step procedure in which weights and corrections are determined simultaneously.
Given $I_T$, the
\emph{conditionally} optimal combination weights of  \citet{COW2024} solve the problem
\begin{equation}
   \min_{\vw} \E[L(e_{c,T+h | T})|I_T]. \label{eq:opt_problem}
\end{equation}
For the conditional expectation in the optimization (\ref{eq:opt_problem}) to be well defined, one must specify how the conditional information is used.
While \citet{COW2024} used $I_T$, the information set available at time $T$, to extract the conditional bias of individual models producing the conditionally optimal weights, we will focus on the combined error.

We assume that the combined forecast error, $e_{c,T+h | T}$, can be decomposed into two parts, $e_{c,T+h | T} = b_{T} +\xi_{T+h}$, where $b_{T}=\E(e_{c,T+h | T}|I_T)$ and $\E(\xi_{T+h}|I_T)=0$. A simple example where this decomposition arises naturally is when forecasting errors exhibit autocorrelation, as in our motivating example in the introduction.
In this case, the information set $I_T$ consists of previous forecasting errors, i.e., $I_T = \{e_{c,T | T-1}\}$, and the conditional bias follows a first-order autoregressive model $e_{c,T+1 | T} = \gamma e_{c,T | T-1} + \xi_{T+1}$ (for $h=1$), with $b_{T}=\gamma e_{c,T | T-1}$ and the remainder $\xi_{T+1}=e_{c,T+1 | T}-\gamma e_{c,T | T-1}$.

We now approach the optimal weight problem by noting that, as \citet{granger1984improved} observed, the unconditional optimal weights can be estimated using ordinary least squares (OLS) in the linear regression
\begin{equation}
   y_{t+h} = w_1 f_{1,t+h | t} + \dots + w_n f_{n,t+h | t} +  e_{c,t+h | t},\; t = 1,\dots,T-h \label{eq:OLS}
\end{equation}
subject to the restriction $w_1 + \dots + w_n =1$. The implicit assumption is that $e_{c,t+h | t}$ are independent and identically distributed (\emph{iid}). If the errors are correlated, i.e., the \emph{iid} assumption does not hold, which is rather common as noted by \citet{Diebold1988}, then the OLS estimator, $\widehat{\vw}^{\text{OLS}}$ (and its corresponding forecast $f_{T+h | T}^{\text{OLS}}$), remains unbiased and consistent but is not necessarily efficient.
{In the following GLS analysis, we consider a fixed forecast horizon $h$
and suppress the horizon dependence of the combination weights and
correction parameters.}
We can therefore estimate the corrected forecast combination using
\begin{equation*}
   y_{t+h} = w_1 f_{1,t+h | t} + \dots + w_n f_{n,t+h | t} +  \gamma \left[y_t - (w_1 f_{1,t | t-h} + \dots + w_n f_{n,t | t-h}) \right]+ \xi_{t+h},\; 
\end{equation*}
which is equivalent to
\begin{equation*}
   y_{t+h} - \gamma y_t = w_1 (f_{1,t+h | t} - \gamma f_{1,t | t-h}) + \dots + w_n (f_{n,t+h | t} - \gamma f_{n,t | t-h}) + \xi_{t+h}. 
\end{equation*}
\citet{POOJARI2025100109} refer to this transformation\footnote{{\citet{CoulsonRobins1993} use a version with an intercept as one of their models (Model (2)---Cochrane-Orcutt). \citet{CoulsonRobins1993} models (4) and (5) are also included in our empirical part as additional benchmarks.}} as the Hildreth–Lu method and also present other alternatives that may be more attractive.
\citet{Diebold1988} also considered a general ARMA(p,q) model for $e_{c,t+h | t}$.
We will focus on the simple version and
the corresponding estimators $(\widehat{\vw}^{\text{GLS}},\widehat{\gamma}^{\text{GLS}})$, which can be used to combine and correct at the same time
\begin{equation}
   f_{T+h | T}^{\text{GLS}} = (\widehat{\vw}^{\text{GLS}})'\vf_{T+h | T} + \widehat{\gamma}^{\text{GLS}} (y_T - (\widehat{\vw}^{\text{GLS}})'\vf_{T | T-h} ). \label{eq:GLScorrected}
\end{equation}

When considering the forecasting errors, $e^{\text{OLS}}_{T+h | T} = {y_{T+h}} - f_{T+h | T}^{\text{OLS}}$ and $e_{T+h | T}^{\text{GLS}} = {y_{T+h}} - f_{T+h | T}^{\text{GLS}}$, one naturally expects the GLS to perform better. This result is formalized by the following theorem.

{
\begin{theorem}[Oracle GLS correction]
\label{thm:oracle_corrected_combination}
Let
$
e_{c,T+h\mid T}(\vw)
=
y_{T+h}-\vw'\vf_{T+h\mid T},
\;
e_{c,T\mid T-h}(\vw)
=
y_T-\vw'\vf_{T\mid T-h},
$
with \(\vw'\viota=1\). Suppose that
\(\vf_{T+h\mid T}\), \(\vf_{T\mid T-h}\), and \(y_T\) are
\(I_T\)-measurable and have finite second moments.
Assume that there exist population parameters
\((\vw_0,\gamma_0)\), with \(\vw_0'\viota=1\), such that
\[
y_{T+h}
=
\vw_0'\vf_{T+h\mid T}
+
\gamma_0
\left(
y_T-\vw_0'\vf_{T\mid T-h}
\right)
+
\xi_{T+h},
\]
where
$
\E(\xi_{T+h}\mid I_T)=0$
and
$\E(\xi_{T+h}^2)<\infty.$
Then the oracle corrected combination
\[
f^{\mathrm{C}0}_{T+h\mid T}
=
\vw_0'\vf_{T+h\mid T}
+
\gamma_0 e_{c,T\mid T-h}(\vw_0)
\]
equals the conditional mean:
$
f^{\mathrm{C}0}_{T+h\mid T}
=
\E(y_{T+h}\mid I_T).
$
Consequently, for any \(I_T\)-measurable uncorrected combination
$
f_{c,T+h\mid T}(\vw)
=
\vw'\vf_{T+h\mid T},
\; \vw'\viota=1,
$
we have
\[
\E\!\left[
  \left(
    y_{T+h}-f_{c,T+h\mid T}(\vw)
  \right)^2
  \mid I_T
\right]
-
\E\!\left[
  \left(
    y_{T+h}-f^{\mathrm{C}0}_{T+h\mid T}
  \right)^2
  \mid I_T
\right]
=
\left[
  f^{\mathrm{C}0}_{T+h\mid T}
  -
  f_{c,T+h\mid T}(\vw)
\right]^2
\geq 0
\;\text{a.s.}
\]
and
\[
\E\!\left[
  \left(
    y_{T+h}-f^{\mathrm{C}0}_{T+h\mid T}
  \right)^2
\right]
\leq
\E\!\left[
  \left(
    y_{T+h}-f_{c,T+h\mid T}(\vw)
  \right)^2
\right].
\]
In particular, relative to the uncorrected combination using the
same oracle weights \(\vw_0\), the conditional MSE reduction is
$
\gamma_0^2
\left[
  e_{c,T\mid T-h}(\vw_0)
\right]^2,
$
and the unconditional MSE reduction is
$
\gamma_0^2
\E\!\left[
  e_{c,T\mid T-h}(\vw_0)^2
\right].
$
\end{theorem}}
\begin{proof} See Appendix~\ref{app:proofs}. \end{proof}

{Theorem~\ref{thm:oracle_corrected_combination} is an oracle result
based on the population parameters $(\vw_0,\gamma_0)$ and the correctly
specified conditional-mean model. A feasible corrected combination
constructed using
$(\widehat{\vw}^{\mathrm{GLS}},\widehat{\gamma}^{\mathrm{GLS}})$
does not automatically inherit the conditional or unconditional MSE
ranking established by the theorem. Estimation error in the combination
weights and the serial-correlation parameter, as well as misspecification
of the error dynamics, may offset some or all of the oracle improvement.}

{The preceding theorem focuses on the oracle GLS
correction under squared-error loss. We now provide a more general
result that permits arbitrary admissible correction classes and convex
loss functions. The result extends the information-set comparison in
Theorem~1(c) of \citet{COW2024}: whereas their result compares
conditionally optimal combination weights, we compare nested classes
of corrected combinations, thereby allowing the combination weights
and the correction to be selected jointly.}
{
\begin{theorem}[Oracle dominance of corrected combinations]
\label{thm:oracle_corrected_combinations}
Let $J_T\subseteq I_T$ be two information sets, and suppose that
$\vf_{T+h\mid T}$ is $J_T$-measurable. For
$\mathcal G_T\in\{J_T,I_T\}$, let
$\mathcal C(\mathcal G_T)$ be a nonempty class of admissible
$\mathcal G_T$-measurable corrected combinations of the form
\[
f^{\mathrm C}_{T+h\mid T}
=
\vw'\vf_{T+h\mid T}
+
b_T(\mathcal G_T,\vw,\vgamma),
\qquad
\vw'\viota=1,
\]
where $b_T(\mathcal G_T,\vw,\vgamma)$ is an admissible correction term
based on $\mathcal G_T$. Suppose that
$
\mathcal C(J_T)\subseteq\mathcal C(I_T).
$
Let $L(\cdot)$ be a convex loss function. Assume that the relevant
conditional expectations exist and that the oracle minimizers
$
f^*(\mathcal G_T)
\in
\arg\min_{f\in\mathcal C(\mathcal G_T)}
\E\!\left[
  L(y_{T+h}-f)
  \mid\mathcal G_T
\right]
$
exist\footnote{{Here, $f^*(\mathcal G_T)$ is the oracle minimizer within the
restricted class $\mathcal C(\mathcal G_T)$ and need not equal
$\E(y_{T+h}\mid\mathcal G_T)$. Under squared-error loss, the two
coincide if and only if the conditional mean belongs to
$\mathcal C(\mathcal G_T)$.}}
for $\mathcal G_T=J_T$ and $\mathcal G_T=I_T$. Then
\[
\E\!\left[
  \left.
  \min_{f\in\mathcal C(I_T)}
  \E\!\left[
    L(y_{T+h}-f)\mid I_T
  \right]
  \right|J_T
\right]
\leq
\min_{f\in\mathcal C(J_T)}
\E\!\left[
  L(y_{T+h}-f)\mid J_T
\right].
\]
Equivalently,
$
\E\!\left[
  L\!\left(y_{T+h}-f^*(I_T)\right)
  \mid J_T
\right]
\leq
\E\!\left[
  L\!\left(y_{T+h}-f^*(J_T)\right)
  \mid J_T
\right]
$
and
$
\E\!\left[
  L\!\left(y_{T+h}-f^*(I_T)\right)
\right]
\leq
\E\!\left[
  L\!\left(y_{T+h}-f^*(J_T)\right)
\right].
$
\end{theorem}
\begin{proof} See Appendix~\ref{app:proofs}. \end{proof}
}

{
The OLS and GLS combinations above provide a special case of
Theorem~\ref{thm:oracle_corrected_combinations}. Let $J_T$ contain
the information used to construct the current component forecasts
$\vf_{T+h\mid T}$, and augment it with the latest realization and the
corresponding lagged component forecasts:
\[
I_T
=
J_T\vee\sigma\!\left(y_T,\vf_{T\mid T-h}\right).
\]
The uncorrected combination class associated with the restricted OLS
regression~(\ref{eq:OLS}) is
\[
\mathcal C(J_T)
=
\left\{
  \vw'\vf_{T+h\mid T}:
  \vw'\viota=1
\right\}.
\]
The class of combinations corrected using the lagged error
$e_{c,T\mid T-h}$ under the restricted GLS
specification~(\ref{eq:GLScorrected}) is
\[
\mathcal C(I_T)
=
\left\{
  \vw'\vf_{T+h\mid T}
  +
  \gamma\left(y_T-\vw'\vf_{T\mid T-h}\right):
  \vw'\viota=1,\ \gamma\in\Gamma
\right\}.
\]
If $0\in\Gamma$, setting $\gamma=0$ recovers every uncorrected
combination, and hence
$\mathcal C(J_T)\subseteq\mathcal C(I_T)$. Theorem~\ref{thm:oracle_corrected_combinations}
therefore implies that the oracle minimum conditional expected loss in
the corrected class, averaged conditional on $J_T$, cannot exceed the
oracle minimum conditional expected loss in the uncorrected class.
Restricted OLS and feasible GLS provide sample analogues of these two
classes, respectively. The oracle ordering does not, however, imply
that the estimated feasible GLS forecast must outperform the estimated
OLS forecast in a finite sample, because estimation and specification
errors are absent from the oracle comparison.
}

\subsection{Generalization} \label{sec:generalization}

{Having generalized the oracle comparison across correction classes and
loss functions, we now generalize the time-series specification of the
combined forecast errors.}
Let $\vy = (y_{1+h},\dots,y_{T})'$, $\mathbf{F} = (\vf_{1+h|1},\dots,\vf_{T|T-h})'$, $\vf_{t+h|t} = (f_{1,t+h|t},\dots,f_{n,t+h|t})'$, and $\ve_{c} = (e_{c,1+h|1},\dots,e_{c,T|T-h})'$. In matrix notation, we write the model (\ref{eq:OLS}) as $\vy = \mathbf{F}\vw + \ve_{c}$. The time-series dependence in the forecasting errors can be modeled with a general covariance matrix $\E[\ve_{c}\ve_{c}']= \Omega(\vgamma)$. Suppose that $\ve_{c}$ follows a stationary ARMA$(p,q)$ process, similar to \citet{Diebold1988}, and that $\Omega(\vgamma)$ is known. Then we solve the following optimization problem
\begin{equation}
   \min_{\vw} (\vy - \mathbf{F} \vw)' \Omega^{-1}(\vgamma) (\vy - \mathbf{F} \vw) \label{eq:opt_problem_GLS}
\end{equation}
subject to the restriction $\vw'\viota=1$, which produces the optimal weights
\begin{equation}
\vw^{\text{opt}} = \frac{\widetilde{\Sigma}^{-1}\viota}{\viota'\widetilde{\Sigma}^{-1}\viota}, \label{eq:opt_weights}
\end{equation}
where $\widetilde{\Sigma}=\widetilde{\mathbf E}' \,\widetilde{\mathbf E} $, $\widetilde{\mathbf E}  = \Omega^{-1/2}(\vgamma) \mathbf E$, and $\mathbf E  = (\vy\viota'-\mathbf{F})$.

In practice, $\Omega(\vgamma)$ is unknown, but the efficient generalized least squares $(\widehat{\vw}^{\text{GLS}},\widehat{\vgamma}^{\text{GLS}})$ can be estimated simultaneously using the maximum likelihood estimation (MLE) method proposed by \citet{harvey1979maximum} or a two-step procedure proposed by \citet{zinde1991estimation}. Let $R(\vw)=(\vy - \mathbf{F}\vw)'\Omega^{-1}(\vgamma) (\vy - \mathbf{F}\vw)$ and $R_T(\vw)=(\vy - \mathbf{F}\vw)'\Omega^{-1}(\widehat{\vgamma}^{\text{GLS}})(\vy - \mathbf{F}\vw)$ denote the theoretical and empirical risks, respectively. They differ only in the use of the true parameter $\vgamma$ versus its estimator $\widehat{\vgamma}^{\text{GLS}}$. 
{The following theorem provides a deterministic perturbation bound for the difference between $R(\vw^\text{opt})$ and $R_T(\widehat{\vw}^\text{GLS})$.}

\begin{theorem} \label{Th4}
   Let $a_{T}=\|\mathbf E'(\Omega^{-1}(\widehat{\vgamma}^{\text{GLS}})-\Omega^{-1}(\vgamma))\mathbf E\|_{\infty}$.
   {If $\|\vw^\text{opt}\|_1\leq{c}$ and $\|\widehat{\vw}^\text{GLS}\|_1\leq{c}$,} then, we have
   \begin{align*}
   |R(\vw^\text{opt})-R_T(\widehat{\vw}^\text{GLS})|&\leq a_{T}c^2,\\
   |R(\widehat{\vw}^\text{GLS})-R_T(\widehat{\vw}^\text{GLS})|&\leq a_{T}c^2,\\
   |R(\widehat{\vw}^\text{GLS})-R(\vw^\text{opt})|&\leq 2a_{T}c^2.
   \end{align*}
\end{theorem}
\begin{proof} See Appendix~\ref{app:proofs}. \end{proof}

{Theorem \ref{Th4} shows that the theoretical minimum risk $R(\vw^\text{opt})$ and the empirical GLS risk $R_T(\widehat{\vw}^\text{GLS})$ are approximately equal provided that $\Omega^{-1}(\widehat{\vgamma}^{\text{GLS}})$ is sufficiently close to $\Omega^{-1}(\vgamma)$ so that $a_{T}c^2=o_{p}(1)$. Note that $\|\vw\|_1$ might not be bounded under the restriction $\vw'\viota=1$.\footnote{{The boundedness of $\|\vw\|_1$ can be established under a standard eigenvalue condition. Suppose that $0<\lambda_{1}\leq\lambda_{\min}(T^{-1}\widetilde{\Sigma})
\leq\lambda_{\max}(T^{-1}\widetilde{\Sigma})\leq\lambda_{2}<\infty$, where $\widetilde{\Sigma}=\mathbf E'\Omega^{-1}(\vgamma)\mathbf E$. Then it follows that
$\|\vw^\text{opt}\|_1\leq \sqrt{n}\|\vw^\text{opt}\|_2
=\sqrt{n}\frac{\|\widetilde{\Sigma}^{-1}\viota\|_2}{\viota'\widetilde{\Sigma}^{-1}\viota}
\leq\frac{\lambda_{\max}(T^{-1}\widetilde{\Sigma})}{\lambda_{\min}(T^{-1}\widetilde{\Sigma})}
\leq\frac{\lambda_2}{\lambda_1}\equiv{c}<\infty$. By a similar argument, $\|\widehat{\vw}^\text{GLS}\|_1$ is bounded under the corresponding eigenvalue condition. Furthermore, we have $\|\vw\|_1=1$ when the weights are nonnegative.}} In practice, we can impose the restriction $\|\vw\|_1\leq{c}$ in the optimization problem (\ref{eq:opt_problem_GLS}) to ensure that the weights are bounded.}

{To make the resulting combined forecast explicit, define the covariance between the historical combined errors and the future combined error as
$
   \vdelta_T(\vgamma)=\E[\ve_c\, e_{c,T+h|T}]
$
and the variance $\sigma_{T+h}^{2}:=\operatorname{Var}(e_{c,T+h\mid T}).$
The joint covariance matrix can then be partitioned as
\begin{equation*}
 \E\left[
 \begin{pmatrix}\ve_c\\ e_{c,T+h|T}\end{pmatrix}
 \begin{pmatrix}\ve_c\\ e_{c,T+h|T}\end{pmatrix}'
 \right]
 =
 \begin{pmatrix}
    \Omega(\vgamma) & \vdelta_T(\vgamma)\\
    \vdelta_T'(\vgamma) & \sigma_{T+h}^2
 \end{pmatrix}.
\end{equation*}
For a Gaussian error process, the conditional forecast of the future combined error is
\begin{equation}
   b_T^{\mathrm{GLS}}(\vgamma)
   =\E[e_{c,T+h|T}\mid\ve_c]
   =\vdelta_T'(\vgamma)\Omega^{-1}(\vgamma)\ve_c. \label{eq:general_linear_correction}
\end{equation}
Without Gaussianity, the right-hand side of~(\ref{eq:general_linear_correction}) remains the best linear predictor of the future error under MSE loss. Substituting the GLS estimates and \(\widehat{\ve}_c=\vy-\mathbf{F}\widehat{\vw}^{\mathrm{GLS}}\) gives the feasible corrected combined forecast
\begin{equation}
   f_{T+h|T}^{\mathrm{GLS}}
   =(\widehat{\vw}^{\mathrm{GLS}})'\vf_{T+h|T}
   +\vdelta_T'(\widehat{\vgamma}^{\mathrm{GLS}})
    \Omega^{-1}(\widehat{\vgamma}^{\mathrm{GLS}})
    (\vy-\mathbf{F}\widehat{\vw}^{\mathrm{GLS}}). \label{eq:general_GLS_forecast}
\end{equation}
For an ARMA$(p,q)$ process, the second term in~(\ref{eq:general_GLS_forecast}) is the usual forecast of the future ARMA error based on the errors observed by time~$T$. For $h=1$, equation~(\ref{eq:GLScorrected}) is its AR(1) special case.}

{We have considered the class of additive-linear corrections $\vw'\vf_{T+h|T} + b_T$.
However, the linearity of $b_T^{\mathrm{GLS}}(\vgamma)$ is a useful simplification, while the oracle optimal correction term $b_T(I_T,\vgamma)=\E[e_{c,T+h|T}\mid I_T]$, given by Theorem~\ref{Th1_combined},  can be non-linear. One can straightforwardly adapt our approach for non-linear $b_T(I_T,\vgamma)$ as long as $I_T$ is available and $\vgamma$ can be estimated. Moreover, in some situations, one might even explore multiplicative  $\vw'\vf_{T+h|T}\times b_T(I_T,\vgamma)$ or more general adjustments.}

While our theory allows for general correction, parsimonious models often deliver good forecasts in practice. Therefore, the practical advice is to use simple corrections, as in our motivating example in the introduction, before adding complexity and estimation variability that could harm forecasting performance.

\subsection{Simulation evidence}

To assess when the correction of the combined forecast error is beneficial, we conduct a simple Monte Carlo experiment. The design considers two individual one-step-ahead forecast errors, denoted by $e_{1,t+1 | t}$ and $e_{2,t+1 | t}$, each generated as a stationary AR(1) process,
\begin{align*}
    e_{1,t+1 | t} &= \rho_1 e_{1,t | t-1} + \varepsilon_{1,t+1}, \\
    e_{2,t+1 | t} &= \rho_2 e_{2,t | t-1} + \varepsilon_{2,t+1},
\end{align*}
where $\varepsilon_{i,t} \sim N(0,\sigma_i^2)$, $i=1,2$, uncorrelated, independent and identically distributed. In the simulation, we set $\sigma_1=\sigma_2=1$, and the initial values are drawn from the stationary distribution.
The sample size is set to $T=120$, corresponding to ten years of monthly forecasts.

For each pair $(\rho_1,\rho_2)$, the equal-weight combined forecast error is constructed as
$
    e^{\text{EW}}_{t+1 | t} = (e_{1,t+1 | t} + e_{2,t+1 | t})/2.
$
We then compare this benchmark with two corrected versions. The first uses a fixed correction parameter $\gamma=0.5$:
$
    e^{\text{CEW1}}_{t+1 | t}  = e^{\text{EW}}_{t+1 | t}  - \gamma e^{\text{EW}}_{t | t-1}.
$
The second corrected version uses a historically optimal correction parameter (\ref{eq:optCFh}):
$
    e^{\text{CEW2}}_{t+1 | t}
    =
    e^{\text{EW}}_{t+1 | t}  - \widehat{\gamma}_t e^{\text{EW}}_{t | t-1}.
$
{Although $e_{1,t+1\mid t}$ and $e_{2,t+1\mid t}$ follow AR(1) processes, Theorem~\ref{thm:combined-error-arma21} in Appendix~\ref{app:proofs} shows that their combination admits the ARMA$(2,1)$ representation\footnote{We thank Wilson Chen for drawing our attention to this property during a discussion.} in~(\ref{eq:combined-error-arma21}), which is minimal except under special parameter restrictions. Thus, the combined error is generally not AR(1), and the one-lag AR corrections $e^{\text{CEW1}}_{t+1 | t}$ and $e^{\text{CEW2}}_{t+1 | t}$ considered here are misspecified.}

The experiment is repeated over a grid
$
    \rho_1,\rho_2 \in \{0,0.05,0.10,\ldots,0.90\}.
$
For each grid point, we use $5{,}000$ Monte Carlo replications. The performance of each corrected combination is evaluated by its MSFE relative to the uncorrected equal-weight combination.
Values below one indicate that the correction improves upon the uncorrected equal-weight combination.

\begin{figure}[htbp]
    \centering

    \begin{subfigure}{0.45\textwidth}
        \centering
        \includegraphics[width=\linewidth]{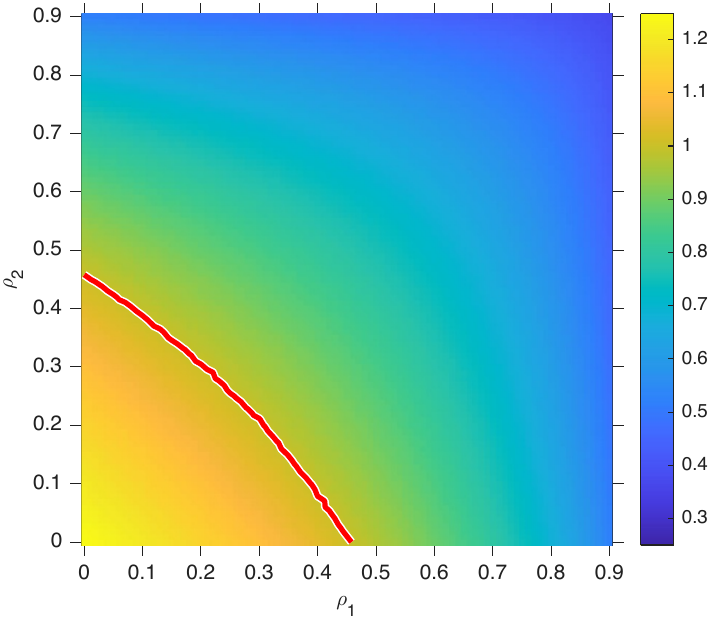}
        \caption{CEW1 with fixed $\gamma = 0.5$}
        \label{fig:simul-gamma-05}
    \end{subfigure}
    \hfill
    \begin{subfigure}{0.45\textwidth}
        \centering
        \includegraphics[width=\linewidth]{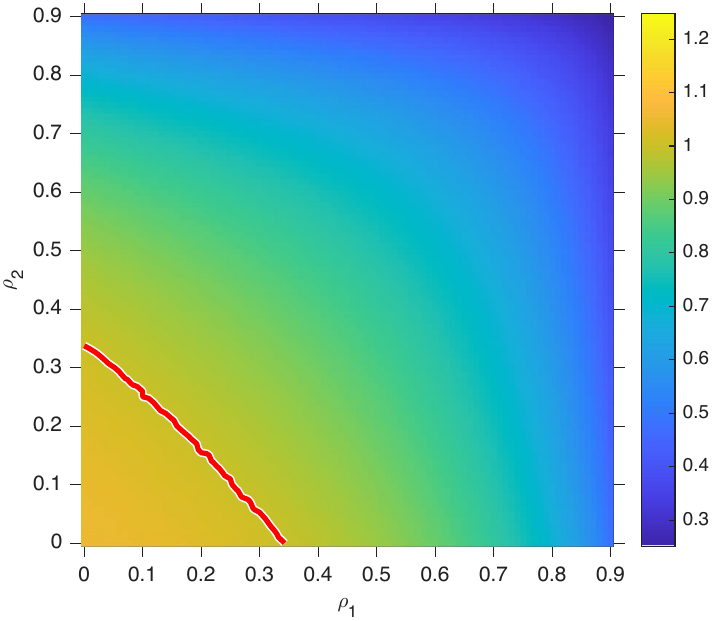}
        \caption{CEW2 with historically optimal $\widehat{\gamma}_t$}
        \label{fig:simul-gamma-ho}
    \end{subfigure}

    \caption{Heat maps of the relative MSFE of corrected equal-weight combination with (a) fixed $\gamma = 0.5$ and (b) historically optimal $\widehat{\gamma}_t$ as a function of $(\rho_1,\rho_2)$. The red line indicates the relative MSFE = 1. The correction is beneficial in the region above and to the right of the red line.}
    \label{fig:simul}
\end{figure}

 Figure~\ref{fig:simul} maps out the regions of persistence in the individual forecast errors for which correcting the combined forecast error is expected to be beneficial: the large top right area above the red line. The bottom left corner below the red line is the area where the autocorrelation of original forecast errors is small, and therefore, the correction is not beneficial: estimation of the historically optimal $\widehat{\gamma}_t$ introduces more noise than benefits, while the use of fixed $\gamma = 0.5$ introduces more correlation than it fixes. The simulation also confirms that correction with the fixed $\gamma = 0.5$ derives the majority of benefits already, while the historically optimal $\widehat{\gamma}_t$ extracts additional benefits for a narrow band of moderate $(\rho_1,\rho_2)$ in this stylized setting. For practical purposes, where even one forecast exhibits a strong correlation pattern (e.g., the ES forecast with the autocorrelation of 0.77 in our motivating example), the correction of the combined forecast is highly advisable.

{To examine robustness to unstable error dynamics, we
augment the preceding experiment with three structural-break designs.
We retain the sample size $T=120$, the $5{,}000$ Monte Carlo
replications, and mutually independent Gaussian innovations with
$\sigma_1=\sigma_2=1$, changing only the persistence parameters at the
sample midpoint. Specifically, after observation 60,
$(\rho_1,\rho_2)$ changes from $(0.8,0.6)$ to $(0.2,0.1)$ in the
``Persistence falls'' design, from $(0.2,0.1)$ to $(0.8,0.6)$ in the
``Persistence rises'' design, and from $(0.8,0.6)$ to $(0.8,0.1)$ in
the ``One process changes'' design. Initial values are drawn from the
stationary distribution associated with the first regime. At the
break, only the persistence parameters change: the last pre-break
errors remain the initial conditions for the post-break processes, so
$e_{i,61\mid60}=\rho_i^{\mathrm{after}}e_{i,60\mid59}
+\varepsilon_{i,61}$. The processes therefore transition naturally
into the new regime rather than being restarted from its stationary
distribution. We apply the same two
corrections as in the preceding experiment: CEW1 uses the fixed
$\gamma=0.5$, while CEW2 estimates $\widehat{\gamma}_t$ recursively
from an expanding window. Neither procedure is informed about the
timing or magnitude of the break, and the expanding window is not
reset after it occurs. Table~\ref{tab:simul-breaks} reports relative
MSFE over the full sample and the post-break sample, with the latter
divided into the First 12 periods and the Remaining periods.}

\begin{table}[htbp]
\centering

\caption{Relative MSFE under breaks in the individual error dynamics}
\label{tab:simul-breaks}
\small
\begin{tabular}{lcccccc}
\toprule
\multicolumn{7}{l}{\textit{Panel A: Full-sample and post-break performance}}\\
& & & \multicolumn{2}{c}{Full sample}
& \multicolumn{2}{c}{Post-break sample}\\
\cmidrule(lr){4-5}\cmidrule(lr){6-7}
Scenario & Before & After & Fixed & Hist. & Fixed & Hist.\\
\midrule
Persistence falls
& $(0.8,0.6)$ & $(0.2,0.1)$ & 0.7302 & 0.7440 & 1.1076 & 1.2099\\
Persistence rises
& $(0.2,0.1)$ & $(0.8,0.6)$ & 0.7322 & 0.7551 & 0.5538 & 0.6222\\
One process changes
& $(0.8,0.6)$ & $(0.8,0.1)$ & 0.5941 & 0.5836 & 0.6692 & 0.6788\\
\midrule
\multicolumn{7}{l}{\textit{Panel B: Post-break adaptation}}\\
& & & \multicolumn{2}{c}{First 12 periods}
& \multicolumn{2}{c}{Remaining periods}\\
\cmidrule(lr){4-5}\cmidrule(lr){6-7}
Scenario & Before & After & Fixed & Hist. & Fixed & Hist.\\
\midrule
Persistence falls
& $(0.8,0.6)$ & $(0.2,0.1)$ & 1.1375 & 1.3464 & 1.1075 & 1.1868\\
Persistence rises
& $(0.2,0.1)$ & $(0.8,0.6)$ & 0.6443 & 0.7914 & 0.5543 & 0.5968\\
One process changes
& $(0.8,0.6)$ & $(0.8,0.1)$ & 0.7497 & 0.8029 & 0.6744 & 0.6837\\
\bottomrule
\end{tabular}

\vspace{2pt}
\begin{minipage}{0.98\textwidth}
\footnotesize
\textit{Notes:} ``Before'' and ``After'' report $(\rho_1,\rho_2)$
on either side of the midpoint break. Fixed denotes CEW1 with
$\gamma=0.5$, and Hist. denotes CEW2 with the recursively estimated
historical correction coefficient. Each entry is the Monte Carlo
average, over $5{,}000$ replications, of the MSFE relative to the
uncorrected equal-weight combination over the indicated evaluation
window. Values below one favor correction.
\end{minipage}
\end{table}

{When persistence falls, both corrections remain
beneficial over the full sample but become harmful after the break.
When persistence rises, both corrections produce substantial
post-break gains, although the recursively estimated coefficient
adjusts only gradually. When only the second error process becomes
less persistent, both corrections continue to improve on the
uncorrected combination. The comparison between the ``First 12
periods'' and the ``Remaining periods'' in Panel~B shows that the
break harms correction performance initially: in every scenario, the
relative MSFE is lower over the Remaining periods than over the First
12 periods. Once the First 12 periods after the break are excluded,
substantial correction gains are reinstated when persistence rises or
only one process becomes less persistent. When persistence falls for
both processes, excluding the First 12 periods reduces the damage from
the break, but does not restore the gains from correction because the
new regime contains too little persistence for either correction to
be beneficial.}

\section{Empirical illustration: The US Survey of Professional Forecasters} \label{sec:empirical}

The analysis uses individual-level forecasts from the U.S. Survey of Professional Forecasters (SPF), conducted quarterly by the Federal Reserve Bank of Philadelphia. The SPF collects forecasts from a panel of professional economists for a range of key U.S. macroeconomic variables across multiple horizons. We focus on one-step-ahead forecasts ($h=1$) and four major macroeconomic indicators: the civilian unemployment rate (UNEMP), real gross domestic product (RGDP), industrial production growth (INDPROD), and consumer price inflation (CPI). The data were accessed in October 2025\footnote{\url{https://www.philadelphiafed.org/surveys-and-data/data-files}}.
Figure~\ref{fig:macro} presents a graphical overview of the data.
The correlation in SPF forecast errors is significant and can be driven, {for example, by departures from full-information rational expectations. \citet{coibion2015} document information rigidity in consensus forecasts, while \citet{bordalo2020} show that individual forecasters typically overreact to news even though consensus forecasts underreact.}

\begin{figure}[htbp]
    \centering

    \begin{subfigure}{0.45\textwidth}
        \centering
        \includegraphics[width=\linewidth]{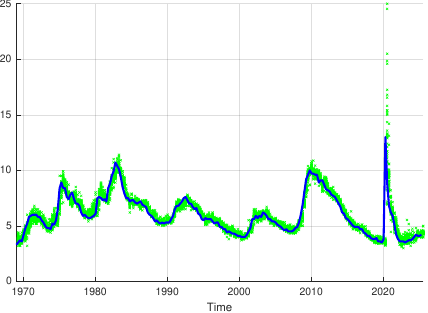}
        \caption{UNEMP}
        \label{fig:UNEMP}
    \end{subfigure}
    \hfill
    \begin{subfigure}{0.45\textwidth}
        \centering
        \includegraphics[width=\linewidth]{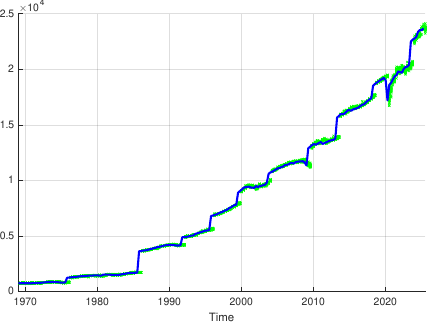}
        \caption{RGDP}
        \label{fig:RGDP}
    \end{subfigure}

    \vspace{0.5cm}

    \begin{subfigure}{0.45\textwidth}
        \centering
        \includegraphics[width=\linewidth]{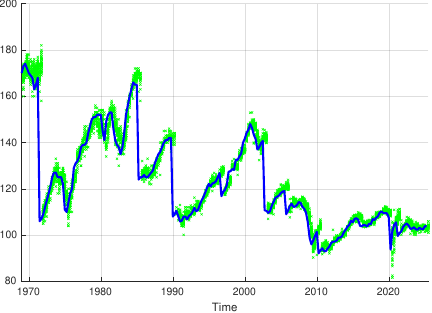}
        \caption{INDPROD}
        \label{fig:INDPROD}
    \end{subfigure}
    \hfill
    \begin{subfigure}{0.45\textwidth}
        \centering
        \includegraphics[width=\linewidth]{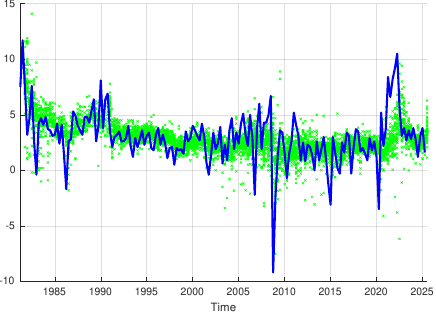}
        \caption{CPI}
        \label{fig:CPI}
    \end{subfigure}

    \caption{Four major macroeconomic indicators considered in the empirical study. The line represents the actual realizations of each macroeconomic variable, while the light crosses represent the individual expert forecasts from the US SPF.}
    \label{fig:macro}
\end{figure}

\subsection{Corrected Mean Forecast}

We first investigate the performance of the corrected mean forecast in our data. Because we have several indicators, we introduce a subscript to distinguish among them, making the correction specific to each indicator
\begin{equation}
   f^{\text{CEW}}_{i,T+1|T} = f^{\text{EW}}_{i,T+1|T} + \gamma_i\, e^{\text{EW}}_{i,T|T-1}, \label{eq:CEW}
\end{equation}
where $i\in$\{UNEMP, RGDP, INDPROD, CPI\}. The forecasts and correction elements remain the same as before, but with the additional subscript $i$.\footnote{{The correction factors are
horizon-specific. We suppress the horizon subscript when it is
clear from context, writing $\gamma_i$ and
$\widehat{\gamma}_{i,t}$ for $\gamma_{i,h}$ and
$\widehat{\gamma}_{i,h,t}$, respectively.}}
To assess the effects of events such as the global financial crisis (GFC) and the COVID pandemic, as well as performance in historical and more recent periods, we consider the periods specified in Table~\ref{tab:periods}.
\begin{table}[ht]
\centering
\begin{tabular}{ll}
\hline
Period                 & Description            \\
\hline
1969Q1--2025Q2         & complete (CPI data starts in 1981Q3)              \\
1969Q1--2025Q2$^*$     & complete (excl. COVID, CPI data starts in 1981Q3)   \\
1969Q1--1999Q4         & before 2000 (CPI data starts in 1981Q3)            \\
2000Q1--2025Q2         & after 2000              \\
2000Q1--2025Q2$^*$     & after 2000 (excl. COVID)  \\
2000Q1--2008Q4         & before GFC              \\
2000Q1--2019Q4         & before COVID            \\
2010Q1--2019Q4         & between GFC and COVID        \\
2022Q1--2025Q2         & after COVID             \\
\hline
\end{tabular}
\caption{Period definitions and descriptions.} \label{tab:periods}
\end{table}

Table~\ref{tab:CorrMeanF} shows the relative root mean squared forecasting errors (RMSFE) across different periods and various correction factors, with the different panes corresponding to UNEMP, RGDP, INDPROD, and CPI.  The results are relative to the
original mean forecast for each row.
Let us consider the UNEMP results in the first pane. The COVID period is characterized by large outliers that are propagated by the correction (\ref{eq:CEW}), as shown in Figure~\ref{fig:UNEMP}. Therefore, when the period 2020-2022 is included in the evaluation, it is better to use the original mean forecast, i.e., all values in the rows corresponding to 1969Q1-2025Q2 and 2000Q1-2025Q2 are greater than one, and the optimal correction parameter is $\gamma_i=0$ (no correction).
However, once this relatively short period is excluded from the evaluation, the optimal correction is in the neighbourhood of 0.5 (used in the motivating example in the introduction) for 1969Q1-2025Q2$^*$, 2000Q1-2025Q2$^*$, and all other periods.
This treatment of the COVID period aligns with the `set to missing' approach of \citet{Hyndman03042025} and is sufficient for our purposes; however, other approaches can also be used.

It appears that a correction factor of 0.5 is a reliable starting point for our application; however, with sufficient historical data, one can use historical data to determine the optimal correction factor via the simple optimization procedure (\ref{eq:optCFh}) introduced earlier, with the additional stabilizing restriction $-1<\widehat{\gamma}_{i,t}<1$.
The final column {(labeled ``recurs. hist.'')} in Table~\ref{tab:CorrMeanF} shows this possibility {and refers to the out-of-sample correction factor optimally estimated using historical data available at the time of the forecast}. As in the evaluation, the COVID period is detrimental to correction, so the values corresponding to COVID are excluded from the sum in (\ref{eq:optCFh}). The correction with the historically optimal correction factor is beneficial in all periods (once COVID is excluded).
Moreover, the estimated corrected factor $\widehat{\gamma}_{i,t}$ delivers the performance close to the optimal fixed correction factor. In fact, the estimated correction factor is very stable for most of the period, ranging between 0.4 and 0.54, as can be seen in Figure~\ref{fig:HistOptCorrFactor} {(blue line for $h=1$).}
\begin{figure}[htbp]
\centering
\includegraphics[width=0.5\linewidth]{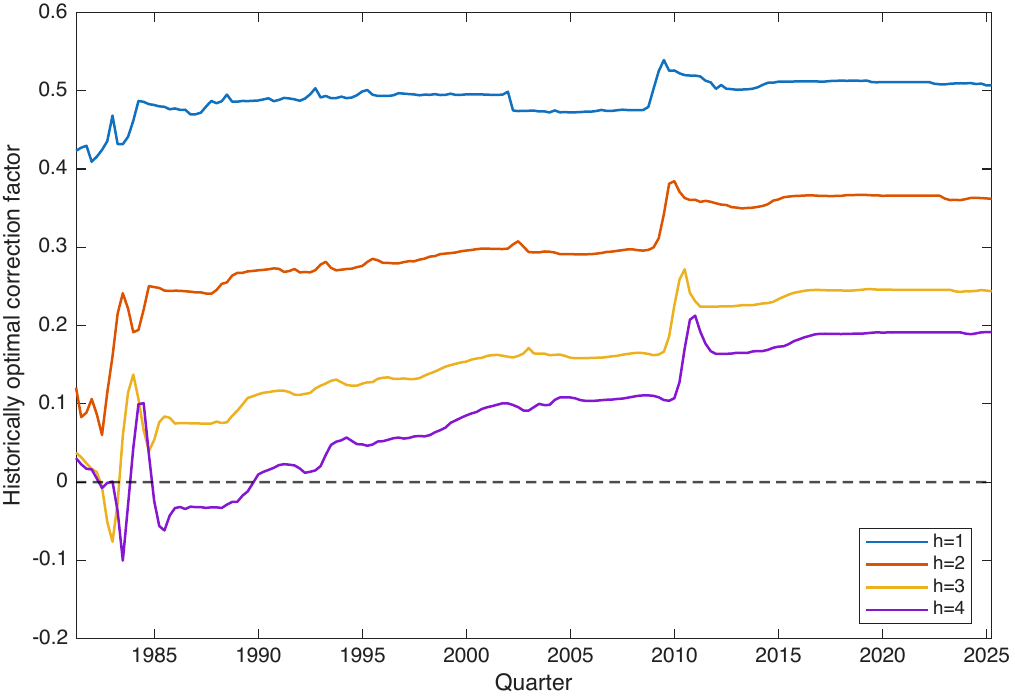}
\caption{Historically optimal correction factors for UNEMP {across multiple horizons} computed using the past available data and used for the mean forecast correction.}
\label{fig:HistOptCorrFactor}
\end{figure}

\begin{table}[H]
\centering \small
\singlespacing
\begin{tabular}{lcccccccccccc}
\hline
 & \multicolumn{10}{c}{Fixed correction factor} & recurs. \\
 \cline{2-11}
Period  & 0.1 & 0.2 & 0.3 & 0.4 & 0.5 & 0.6 & 0.7 & 0.8 & 0.9 & 1 & hist. \\
\hline
\multicolumn{1}{c}{UNEMP}\\
1969Q1--2025Q2  & 1.01 & 1.04 & 1.07 & 1.11 & 1.16 & 1.21 & 1.27 & 1.33 & 1.40 & 1.47 & 1.16 \\
1969Q1--2025Q2$^*$  & 0.95 & 0.91 & 0.88 & 0.86 & \textbf{0.86} & 0.86 & 0.88 & 0.90 & 0.94 & 0.99 & 0.86 \\
1969Q1--1999Q4  & 0.95 & 0.92 & 0.89 & 0.87 & \textbf{0.87} & 0.87 & 0.89 & 0.92 & 0.96 & 1.00 & 0.87 \\
2000Q1--2025Q2  & 1.02 & 1.05 & 1.09 & 1.14 & 1.19 & 1.25 & 1.32 & 1.39 & 1.46 & 1.53 & 1.20 \\
2000Q1--2025Q2$^*$  & 0.95 & 0.90 & 0.86 & 0.84 & \textbf{0.83} & 0.83 & 0.84 & 0.87 & 0.90 & 0.95  & 0.83 \\
2000Q1--2008Q4  & 0.97 & 0.94 & 0.92 & 0.90 & 0.89 & \textbf{0.89} & 0.89 & 0.91 & 0.92 & 0.95  & 0.90 \\
2000Q1--2019Q4  & 0.95 & 0.90 & 0.87 & 0.85 & 0.83 & \textbf{0.83} & 0.84 & 0.87 & 0.90 & 0.94  & 0.84 \\
2010Q1--2019Q4  & 0.97 & 0.95 & 0.94 & \textbf{0.93} & 0.94 & 0.96 & 0.99 & 1.03 & 1.08 & 1.13  & 0.95 \\
2022Q1--2025Q2  & 0.91 & 0.84 & 0.79 & 0.75 & \textbf{0.74} & 0.75 & 0.79 & 0.85 & 0.92 & 1.00 &  0.74 \\
\hline
\multicolumn{1}{c}{RGDP}\\
1969Q1--2025Q2     & 0.96 & 0.93 & 0.91 & \textbf{0.90} & 0.90 & 0.92 & 0.94 & 0.97 & 1.01 & 1.06 & 0.94 \\
1969Q1--2025Q2$^*$ & 0.95 & 0.91 & 0.88 & 0.86 & \textbf{0.85} & 0.86 & 0.87 & 0.90 & 0.93 & 0.98 & 0.90 \\
1969Q1--1999Q4     & 0.95 & 0.91 & 0.88 & 0.86 & 0.85 & \textbf{0.85} & 0.86 & 0.88 & 0.91 & 0.95 & 0.95 \\
2000Q1--2025Q2     & 0.96 & 0.94 & 0.92 & \textbf{0.92} & 0.92 & 0.94 & 0.97 & 1.01 & 1.05 & 1.11 & 0.94 \\
2000Q1--2025Q2$^*$ & 0.95 & 0.91 & 0.88 & 0.86 & \textbf{0.86} & 0.86 & 0.88 & 0.91 & 0.95 & 0.99 & 0.87 \\
2000Q1--2008Q4     & 0.95 & 0.91 & \textbf{0.89} & 0.89 & 0.91 & 0.95 & 1.01 & 1.08 & 1.16 & 1.25 & 0.94 \\
2000Q1--2019Q4     & 0.95 & 0.92 & 0.89 & 0.87 & \textbf{0.87} & 0.88 & 0.90 & 0.93 & 0.97 & 1.02 & 0.89 \\
2010Q1--2019Q4     & 0.95 & 0.91 & 0.88 & 0.86 & \textbf{0.85} & 0.85 & 0.87 & 0.89 & 0.93 & 0.97 & 0.87 \\
2022Q1--2025Q2     & 0.95 & 0.90 & 0.87 & 0.85 & 0.83 & \textbf{0.83} & 0.84 & 0.87 & 0.90 & 0.94 & 0.85 \\
\hline
\multicolumn{1}{c}{INDPROD}\\
1969Q1--2025Q2     & 0.95 & 0.91 & 0.88 & 0.86 & \textbf{0.85} & 0.86 & 0.87 & 0.90 & 0.93 & 0.98 & 0.86 \\
1969Q1--2025Q2$^*$ & 0.95 & 0.91 & 0.88 & 0.86 & \textbf{0.85} & 0.85 & 0.86 & 0.89 & 0.92 & 0.97 & 0.85 \\
1969Q1--1999Q4     & 0.95 & 0.91 & 0.88 & 0.86 & \textbf{0.85} & 0.85 & 0.87 & 0.89 & 0.93 & 0.97 & 0.86 \\
2000Q1--2025Q2     & 0.95 & 0.92 & 0.89 & 0.88 & \textbf{0.87} & 0.88 & 0.90 & 0.92 & 0.96 & 1.01 & 0.88 \\
2000Q1--2025Q2$^*$ & 0.95 & 0.91 & 0.87 & 0.85 & \textbf{0.84} & 0.84 & 0.85 & 0.87 & 0.91 & 0.95 & 0.84 \\
2000Q1--2008Q4     & 0.95 & 0.91 & 0.87 & 0.85 & 0.84 & \textbf{0.83} & 0.85 & 0.87 & 0.90 & 0.94 & 0.84 \\
2000Q1--2019Q4     & 0.95 & 0.90 & 0.87 & 0.85 & 0.84 & \textbf{0.84} & 0.85 & 0.87 & 0.90 & 0.95 & 0.84 \\
2010Q1--2019Q4     & 0.96 & 0.93 & 0.90 & 0.89 & \textbf{0.89} & 0.90 & 0.92 & 0.96 & 1.00 & 1.05 & 0.89 \\
2022Q1--2025Q2     & 0.97 & 0.96 & \textbf{0.95} & 0.96 & 0.97 & 1.00 & 1.03 & 1.08 & 1.13 & 1.18 & 0.98 \\
\hline
\multicolumn{1}{c}{CPI}\\
1981Q3--2025Q2     & 0.98 & 0.96 & \textbf{0.96} & 0.97 & 0.99 & 1.01 & 1.05 & 1.09 & 1.14 & 1.20 & 0.98 \\
1981Q3--2025Q2$^*$ & 0.98 & 0.97 & \textbf{0.97} & 0.98 & 1.00 & 1.03 & 1.07 & 1.12 & 1.17 & 1.23 & 0.98 \\
1981Q3--1999Q4     & 0.98 & 0.97 & \textbf{0.97} & 0.98 & 1.00 & 1.03 & 1.07 & 1.11 & 1.17 & 1.23 & 0.97 \\
2000Q1--2025Q2     & 0.98 & 0.96 & \textbf{0.96} & 0.97 & 0.98 & 1.01 & 1.04 & 1.09 & 1.14 & 1.20 & 0.98 \\
2000Q1--2025Q2$^*$ & 0.98 & 0.97 & \textbf{0.97} & 0.98 & 1.00 & 1.03 & 1.07 & 1.12 & 1.17 & 1.24 & 0.98 \\
2000Q1--2008Q4     & 1.01 & 1.02 & 1.04 & 1.06 & 1.08 & 1.11 & 1.15 & 1.18 & 1.22 & 1.26 & 1.03 \\
2000Q1--2019Q4     & \textbf{0.99} & 0.99 & 1.01 & 1.03 & 1.06 & 1.10 & 1.15 & 1.20 & 1.26 & 1.32 & 1.01 \\
2010Q1--2019Q4     & 0.98 & 0.97 & \textbf{0.96} & 0.97 & 0.99 & 1.02 & 1.06 & 1.11 & 1.16 & 1.22 & 0.97 \\
2022Q1--2025Q2     & 0.92 & 0.85 & 0.79 & 0.74 & 0.70 & \textbf{0.69} & 0.69 & 0.71 & 0.75 & 0.80 & 0.88 \\
\hline
\end{tabular}
\caption{Relative RMSFE for UNEMP, RGDP, INDPROD, and CPI across different periods {at forecast horizon $h=1$} (relative to the original mean forecast). The fixed correction factor is varied from 0 to 1 in increments of 0.1 (0 corresponds to the original mean forecast). The last column shows the result for the {out-of-sample expanding-window recursively estimated} historical correction factor ({as given by (\ref{eq:optCFh}), after an initial warm-up period, using the information available at the time of the forecast,} excluding the COVID period). The minimum value (based on the full precision) that improves on the mean forecast in each row (across the fixed factor) is bold.\\
Note: $^*$ indicates that the COVID period (2020Q1-2022Q4) was excluded from evaluation.} \label{tab:CorrMeanF}
\end{table}

The results for RGDP are more pronounced. Even with the COVID period included, we observe around a 10\% improvement from the correction in the 1969Q1--2025Q2 period, and this improvement increases once COVID is excluded from the evaluation in 1969Q1--2025Q2$^*$. In all periods considered, the correction is beneficial, i.e., the minimum relative RMSFE is substantially less than 1 in all rows when a fixed correction factor is used.
Again, a correction factor of around 0.5 delivers good performance in all cases. When using the historically optimal correction factor, as given in the last column, performance drops only marginally (with the largest deterioration in the distant 1969Q1--1999Q4 period), and the corrected mean forecast always outperforms the mean.
The results for INDPROD further strengthen support for the fixed correction factor of 0.5, with only a marginal drop in performance when using the historically optimal correction factor across all periods.

CPI is challenging to forecast, see \citet{COW2024}, and beating simple benchmarks, such as a simple average or random walk variations, by more than a few percent is notoriously difficult.
Keeping this in mind, we can examine the CPI results that diverge from other indicators. The correction yields only minor improvements over the mean forecasts, and the best results are achieved with a correction factor close to 0.3. The important exception is the most recent period after COVID, 2022Q1--2025Q2, where the fixed factor 0.5 yields a 30\% improvement. Again, estimating the historically optimal correction factor affects performance only marginally, except for the recent 2022Q1--2025Q2 period, where the drop is more substantial, yet the final results are 12\% better than the original mean forecast.

\subsubsection*{Multi-step forecasting}

{
The results for different forecast horizons are in Appendix~\ref{app:empirical}. Tables~\ref{tab:CorrMeanF-h2}--\ref{tab:CorrMeanF-h4} extend
Table~\ref{tab:CorrMeanF} to forecast horizons $h=2,3,4$.
Forecasts are aligned by target quarter, and the correction uses the corresponding horizon-specific lagged forecast error. The discussion below focuses on the full period excluding COVID (1969Q1-2025Q2* or 1981Q1-2025Q2* for CPI) and the latest period, 2022Q1--2025Q2.

\noindent\textit{UNEMP.}
For the full period excluding COVID, the minimum relative RMSFE increases from 0.86 at $h=1$ with $\gamma_i=0.5$ (Table~\ref{tab:CorrMeanF}) to 0.93 at $h=2$ with $\gamma_i=0.4$ (Table~\ref{tab:CorrMeanF-h2}), 0.97 at $h=3$ with $\gamma_i=0.3$ (Table~\ref{tab:CorrMeanF-h3}), and 0.98 at $h=4$ with $\gamma_i=0.2$ (Table~\ref{tab:CorrMeanF-h4}). The latest period produces larger but non-monotonic gains: the minimum relative RMSFEs are 0.74, 0.82, 0.71, and 0.98 at horizons 1--4, respectively. The corresponding correction factors are 0.5, 0.5, 0.3, and 0.1. Thus, the full-period evidence shows a steady weakening of the correction as the horizon increases, while the strong latest-period result at $h=3$ is not sustained at $h=4$.

\noindent\textit{RGDP.}
RGDP displays the clearest deterioration across horizons. For the full period excluding COVID, the minimum relative RMSFE rises from 0.85 at $h=1$ with $\gamma_i=0.5$ (Table~\ref{tab:CorrMeanF}) to 0.93 at $h=2$ with $\gamma_i=0.4$ (Table~\ref{tab:CorrMeanF-h2}), 0.96 at $h=3$ with $\gamma_i=0.3$ (Table~\ref{tab:CorrMeanF-h3}), and 0.99 at $h=4$ with $\gamma_i=0.2$ (Table~\ref{tab:CorrMeanF-h4}). The latest-period results follow the same pattern: the minimum relative RMSFE increases from 0.83 to 0.87, 0.93, and 0.97 as the horizon increases, while the optimal correction factor declines from 0.6 to 0.5, 0.4, and 0.2. Hence, for both principal evaluation periods, the benefit of correcting RGDP forecasts becomes progressively smaller at longer horizons.

\noindent\textit{INDPROD.}
For the full period excluding COVID, the minimum relative RMSFE increases monotonically from 0.85 at $h=1$ with $\gamma_i=0.5$ (Table~\ref{tab:CorrMeanF}) to 0.93 at $h=2$ with $\gamma_i=0.4$ (Table~\ref{tab:CorrMeanF-h2}), 0.95 at $h=3$ with $\gamma_i=0.3$ (Table~\ref{tab:CorrMeanF-h3}), and 0.97 at $h=4$ with $\gamma_i=0.2$ (Table~\ref{tab:CorrMeanF-h4}). The latest period is markedly different. The minimum is 0.95 at $h=1$, correction is detrimental at $h=2$ even for the smallest factor considered (a relative RMSFE of 1.07 with $\gamma_i=0.1$), and the minimum is 0.97 at $h=3$. At $h=4$, however, the relative RMSFE falls sharply to 0.75 with $\gamma_i=0.5$. This substantial long-horizon improvement is therefore specific to the latest period and is not present in the full-period results.

\noindent\textit{CPI.}
For the full period excluding COVID, the correction provides only a small improvement at any horizon. The minimum relative RMSFE is 0.97 at $h=1$ with $\gamma_i=0.3$ (Table~\ref{tab:CorrMeanF}), approximately 1.00 at $h=2$ with $\gamma_i=0.1$ (Table~\ref{tab:CorrMeanF-h2}), 0.99 at $h=3$ with $\gamma_i=0.2$ (Table~\ref{tab:CorrMeanF-h3}), and 0.99 at $h=4$ with $\gamma_i=0.1$ (Table~\ref{tab:CorrMeanF-h4}). In contrast, correction is consistently beneficial in the latest period. The minimum relative RMSFEs are 0.69, 0.81, 0.80, and 0.87 at horizons 1--4, obtained with correction factors of 0.6, 0.5, 0.5, and 0.4, respectively. Although the latest-period gain generally weakens with the horizon, it remains economically meaningful even at $h=4$.

Taken together, the results in Table~\ref{tab:CorrMeanF} and Tables~\ref{tab:CorrMeanF-h2}--\ref{tab:CorrMeanF-h4} show that the benefits of forecast-error correction generally weaken as the horizon increases. The RMSFE-minimizing fixed correction factor also tends to decrease, indicating that less of the previous forecast error should be carried forward at longer horizons. The latest-period results are less uniform, with particularly strong exceptions for UNEMP at $h=3$, INDPROD at $h=4$, and CPI throughout the four horizons.

The recursively estimated historical correction factor generally performs worse than the best fixed factor, especially at longer horizons. In the starred full-period rows of Table~\ref{tab:CorrMeanF-h4}, the recursive relative RMSFEs are 1.02 for UNEMP and approximately 1.00 for RGDP, INDPROD, and CPI. For the latest period, the corresponding values are 1.02, 0.98, 0.85, and 0.97, all above the best fixed-factor outcomes for the respective variables. The additional estimation involved in learning the correction factor therefore further reduces the already limited gains at long horizons. Figure~\ref{fig:HistOptCorrFactor} shows that the recursively estimated historical correction becomes smaller as the horizon $h$ increases.
}

\subsubsection*{Timing and real-time availability}

{
The timing of the US SPF is organized around the release of official
macroeconomic data. Since 1990Q3, the survey questionnaire has generally
been sent to participants at the end of the first month of each quarter,
after the US Bureau of Economic Analysis (BEA)  releases its advance
report on the national income and product accounts (NIPA) for the
preceding quarter.
Responses are due in the middle month of the quarter, and the survey results
are published before the BEA's second NIPA estimate \citep[pp.~6--8]{PhiladelphiaFedSPF}. Thus, when submitting
their forecasts in survey quarter $T$, participants normally observe the
first estimate for quarter $T-1$ but not the realization for quarter $T$.

Table~\ref{tab:spf-rgdp-horizons} illustrates the organization
of the quarterly data and is adapted from the official documentation by
\citet[Table~3, p.~22]{PhiladelphiaFedSPF}.
The variable with suffix 1 contains the
historical value for the quarter preceding the survey; suffix 2 contains the
current-quarter nowcast; and suffixes 3--6 contain forecasts for the next
four quarters. Consequently, for a survey conducted in quarter $T$, RGDP1
refers to quarter $T-1$, RGDP2 to quarter $T$, and RGDP3--RGDP6 to quarters
$T+1$ through $T+4$.
}

\begin{table}[H]

\centering
\caption{Forecast horizons for RGDP at three survey dates.
Adapted from the official documentation by \citet[Table~3, p.~22]{PhiladelphiaFedSPF}, with the variable and survey dates modified.
}
\label{tab:spf-rgdp-horizons}

\scriptsize
\singlespacing
\setlength{\tabcolsep}{3pt}
\resizebox{\textwidth}{!}{%
\begin{tabular}{cccccccccc}
\hline
\multicolumn{2}{c}{\makecell{Survey Date\\(Year, Quarter)}} &
\makecell{Quarterly\\Historical\\Value} &
\multicolumn{5}{c}{\makecell{Quarterly Projections:\\Quarter Forecast}} &
\multicolumn{2}{c}{\makecell{Annual-Average\\Projections: Year Forecast}}\\
\cline{1-2}\cline{4-8}\cline{9-10}
\makecell{(1)\\Year} &
\makecell{(2)\\Quarter} &
\makecell{(3)\\RGDP1} &
\makecell{(4)\\RGDP2} &
\makecell{(5)\\RGDP3} &
\makecell{(6)\\RGDP4} &
\makecell{(7)\\RGDP5} &
\makecell{(8)\\RGDP6} &
\makecell{(9)\\RGDPA} &
\makecell{(10)\\RGDPB}\\
\hline
2024 & 3 & 2024:Q2 & 2024:Q3 & 2024:Q4 & 2025:Q1 &
2025:Q2 & 2025:Q3 & 2024 & 2025\\
2024 & 4 & 2024:Q3 & 2024:Q4 & 2025:Q1 & 2025:Q2 &
2025:Q3 & 2025:Q4 & 2024 & 2025\\
2025 & 1 & 2024:Q4 & 2025:Q1 & 2025:Q2 & 2025:Q3 &
2025:Q4 & 2026:Q1 & 2025 & 2026\\
\hline
\end{tabular}%
}

\begin{minipage}{\textwidth}
\footnotesize
\textit{Notes:} RGDP1 is the real-time historical value for the quarter
preceding the survey. RGDP2 is the current-quarter nowcast, and
RGDP3--RGDP6 are forecasts for the following four quarters. RGDPA and
RGDPB are annual-average projections for the current and following years.
\end{minipage}
\end{table}

{
We use the first-release historical values when constructing the forecast
errors. 
This choice keeps the evaluation
close to the real-time information environment faced by the forecasters,
rather than evaluating forecasts against data revised using information that
became available only later. Data revisions may nevertheless affect both the
measured forecast errors and their serial dependence. A systematic comparison
of correction performance across first-release, intermediate, and
latest-available data vintages is an interesting direction for further
research.

For an $h=1$ forecast made in quarter $T$ for target quarter $T+1$, the
immediately preceding forecast error,
\[
e^{\text{EW}}_{i,T\mid T-1}
   = y_{i,T} - f^{\text{EW}}_{i,T\mid T-1},
\]
is not yet observable because $y_{i,T}$ has not been released. If no advance release of $y_{i,T}$ is available, then one can use the nowcast that is reported in the survey. Table~\ref{tab:CorrMeanF-realtime} corresponds to the correction using the mean current-quarter nowcast $f^{\text{EW}}_{i,T\mid T}$ to construct the proxy
\begin{equation}
\widetilde e^{\text{EW}}_{i,T\mid T-1}
   = f^{\text{EW}}_{i,T\mid T}
     - f^{\text{EW}}_{i,T\mid T-1}. \label{eq:err-proxy-nowcast}
\end{equation}
This is a proxy for the preceding forecasting error, rather than the nowcast
error itself. In the starred full-period samples, the proxy correction lowers
the relative RMSFE to 0.95 for UNEMP, whereas the results for RGDP, INDPROD,
and CPI remain close to one. The gains are stronger in the latest period:
the minimum relative RMSFEs are 0.91 for UNEMP and 0.87 for CPI. The
corresponding improvement is negligible for RGDP, while the correction
deteriorates the INDPROD forecast. The recursively estimated factor generally
does not improve upon the best fixed factor.

These results highlight the importance of accurate nowcasts that
reflect current economic conditions. For real GDP,
\citet{giannone2008nowcasting} develop a framework that updates
current-quarter growth estimates as new macroeconomic data are
released. Current applications include the New York Fed Staff
Nowcast\footnote{\url{https://www.newyorkfed.org/research/policy/nowcast}}
and the Now-Casting Economics
service.\footnote{\url{https://www.now-casting.com/}}
Beyond providing a timely assessment of the economy, accurate
nowcasts can support more effective forecast correction by
providing a more reliable proxy for the realization before its
first official release.}

{
Table~\ref{tab:CorrMeanF-realtime-v04a} considers the more conservative
alternative of using only a fully observed preceding-quarter realization.
At survey date $T$, the latest available realized same-horizon error is
\begin{equation}
e^{\text{EW}}_{i,T-1\mid T-2}
   = y_{i,T-1} - f^{\text{EW}}_{i,T-1\mid T-2}, \label{eq:err-extra-lag}
\end{equation}
where $y_{i,T-1}$ denotes the first-release historical value. Because this
error is one target quarter older than the proxy used in Table~\ref{tab:CorrMeanF-realtime}, it may contain less relevant information about
the current forecast error. The results confirm this loss of information. In
the starred full-period samples, the best relative RMSFE is 0.99 for UNEMP
and approximately one, or above one, for the other variables. During
2022Q1--2025Q2, the minimum relative RMSFEs are 0.98 for UNEMP, 0.99 for
RGDP, and 0.88 for CPI, while the correction substantially worsens INDPROD,
whose minimum relative RMSFE is 1.09. The recursively estimated correction
also provides no consistent improvement. Thus, imposing strict real-time
availability weakens the correction gains, although using the contemporaneous
nowcast as a proxy preserves useful information for UNEMP and, especially in
the latest period, CPI.
}

\begin{table}[H]
\centering \footnotesize

\singlespacing
\begin{tabular}{lcccccccccccc}
\hline
 & \multicolumn{10}{c}{Fixed correction factor} & recurs. \\
 \cline{2-11}
Period  & 0.1 & 0.2 & 0.3 & 0.4 & 0.5 & 0.6 & 0.7 & 0.8 & 0.9 & 1 & hist. \\
\hline
\multicolumn{1}{c}{UNEMP}\\
1969Q1--2025Q2  & 1.03 & 1.08 & 1.13 & 1.19 & 1.26 & 1.33 & 1.41 & 1.49 & 1.58 & 1.67 & 1.20 \\
1969Q1--2025Q2$^*$  & 0.98 & 0.96 & 0.95 & \textbf{0.95} & 0.95 & 0.96 & 0.97 & 0.99 & 1.02 & 1.04 & 0.96 \\
1969Q1--1999Q4  & 0.98 & 0.97 & 0.96 & \textbf{0.96} & 0.96 & 0.97 & 0.98 & 1.00 & 1.02 & 1.05 & 0.97 \\
2000Q1--2025Q2  & 1.04 & 1.09 & 1.15 & 1.22 & 1.30 & 1.38 & 1.47 & 1.56 & 1.65 & 1.75 & 1.24 \\
2000Q1--2025Q2$^*$  & 0.98 & 0.96 & 0.94 & 0.94 & \textbf{0.94} & 0.94 & 0.96 & 0.97 & 1.00 & 1.03 & 0.95 \\
2000Q1--2008Q4  & 0.99 & 0.99 & \textbf{0.99} & 0.99 & 0.99 & 1.00 & 1.02 & 1.03 & 1.05 & 1.07 & 0.99 \\
2000Q1--2019Q4  & 0.98 & 0.96 & 0.95 & 0.94 & \textbf{0.94} & 0.95 & 0.96 & 0.97 & 1.00 & 1.02 & 0.95 \\
2010Q1--2019Q4  & 1.00 & 1.01 & 1.02 & 1.03 & 1.06 & 1.08 & 1.11 & 1.15 & 1.19 & 1.23 & 1.05 \\
2022Q1--2025Q2  & 0.96 & 0.94 & 0.92 & \textbf{0.91} & 0.92 & 0.93 & 0.95 & 0.98 & 1.02 & 1.07 & 0.91 \\
\hline
\multicolumn{1}{c}{RGDP}\\
1969Q1--2025Q2     & 1.00 & 1.01 & 1.03 & 1.05 & 1.07 & 1.10 & 1.13 & 1.17 & 1.21 & 1.25 & 1.00 \\
1969Q1--2025Q2$^*$ & \textbf{1.00} & 1.00 & 1.01 & 1.03 & 1.05 & 1.07 & 1.10 & 1.13 & 1.16 & 1.20 & 1.00 \\
1969Q1--1999Q4     & 1.00 & 1.01 & 1.02 & 1.03 & 1.05 & 1.07 & 1.09 & 1.12 & 1.15 & 1.18 & 1.00 \\
2000Q1--2025Q2     & 1.01 & 1.02 & 1.03 & 1.06 & 1.08 & 1.11 & 1.15 & 1.19 & 1.23 & 1.28 & 1.00 \\
2000Q1--2025Q2$^*$ & \textbf{1.00} & 1.00 & 1.01 & 1.03 & 1.05 & 1.07 & 1.10 & 1.13 & 1.16 & 1.20 & 1.00 \\
2000Q1--2008Q4     & \textbf{0.99} & 1.00 & 1.01 & 1.04 & 1.08 & 1.13 & 1.18 & 1.24 & 1.31 & 1.39 & 1.00 \\
2000Q1--2019Q4     & 1.00 & 1.01 & 1.02 & 1.04 & 1.06 & 1.09 & 1.12 & 1.16 & 1.20 & 1.24 & 1.00 \\
2010Q1--2019Q4     & \textbf{1.00} & 1.00 & 1.01 & 1.03 & 1.05 & 1.07 & 1.10 & 1.13 & 1.16 & 1.20 & 1.00 \\
2022Q1--2025Q2     & 1.00 & \textbf{0.99} & 1.00 & 1.01 & 1.02 & 1.04 & 1.06 & 1.08 & 1.11 & 1.14 & 1.00 \\
\hline
\multicolumn{1}{c}{INDPROD}\\
1969Q1--2025Q2     & \textbf{1.00} & 1.00 & 1.01 & 1.03 & 1.04 & 1.07 & 1.09 & 1.12 & 1.16 & 1.20 & 1.02 \\
1969Q1--2025Q2$^*$ & \textbf{1.00} & 1.00 & 1.01 & 1.02 & 1.04 & 1.06 & 1.09 & 1.12 & 1.15 & 1.19 & 1.02 \\
1969Q1--1999Q4     & \textbf{1.00} & 1.00 & 1.01 & 1.03 & 1.04 & 1.07 & 1.09 & 1.13 & 1.16 & 1.20 & 1.03 \\
2000Q1--2025Q2     & \textbf{1.00} & 1.00 & 1.01 & 1.02 & 1.04 & 1.06 & 1.09 & 1.11 & 1.15 & 1.19 & 1.00 \\
2000Q1--2025Q2$^*$ & 0.99 & \textbf{0.99} & 1.00 & 1.01 & 1.02 & 1.04 & 1.06 & 1.09 & 1.12 & 1.15 & 1.00 \\
2000Q1--2008Q4     & 0.99 & \textbf{0.99} & 1.00 & 1.01 & 1.02 & 1.04 & 1.06 & 1.08 & 1.11 & 1.15 & 1.00 \\
2000Q1--2019Q4     & 0.99 & \textbf{0.99} & 1.00 & 1.01 & 1.02 & 1.04 & 1.06 & 1.08 & 1.11 & 1.15 & 1.00 \\
2010Q1--2019Q4     & 1.01 & 1.02 & 1.03 & 1.05 & 1.08 & 1.11 & 1.15 & 1.18 & 1.22 & 1.27 & 1.00 \\
2022Q1--2025Q2     & 1.02 & 1.04 & 1.07 & 1.09 & 1.13 & 1.16 & 1.20 & 1.24 & 1.28 & 1.33 & 1.01 \\
\hline
\multicolumn{1}{c}{CPI}\\
1981Q3--2025Q2     & \textbf{1.00} & 1.00 & 1.00 & 1.01 & 1.02 & 1.03 & 1.04 & 1.06 & 1.08 & 1.10 & 1.01 \\
1981Q3--2025Q2$^*$ & \textbf{1.00} & 1.00 & 1.00 & 1.01 & 1.02 & 1.03 & 1.05 & 1.07 & 1.08 & 1.11 & 1.01 \\
1981Q3--1999Q4     & 1.00 & 1.01 & 1.02 & 1.04 & 1.06 & 1.08 & 1.11 & 1.14 & 1.17 & 1.20 & 1.03 \\
2000Q1--2025Q2     & 1.00 & \textbf{1.00} & 1.00 & 1.00 & 1.01 & 1.01 & 1.02 & 1.04 & 1.05 & 1.07 & 1.00 \\
2000Q1--2025Q2$^*$ & 1.00 & \textbf{1.00} & 1.00 & 1.00 & 1.01 & 1.02 & 1.03 & 1.04 & 1.06 & 1.07 & 1.00 \\
2000Q1--2008Q4     & 1.01 & 1.02 & 1.03 & 1.04 & 1.05 & 1.06 & 1.08 & 1.09 & 1.11 & 1.12 & 1.00 \\
2000Q1--2019Q4     & 1.00 & 1.01 & 1.01 & 1.02 & 1.03 & 1.05 & 1.06 & 1.08 & 1.10 & 1.12 & 1.00 \\
2010Q1--2019Q4     & \textbf{1.00} & 1.00 & 1.00 & 1.01 & 1.02 & 1.03 & 1.04 & 1.06 & 1.08 & 1.10 & 1.00 \\
2022Q1--2025Q2     & 0.97 & 0.95 & 0.93 & 0.91 & 0.90 & 0.88 & 0.87 & 0.87 & \textbf{0.87} & 0.87 & 1.02 \\
\hline
\end{tabular}
\caption{
Relative RMSFE for UNEMP, RGDP, INDPROD, and CPI across different periods at forecast horizon $h=1$, using \textbf{the proxy error (\ref{eq:err-proxy-nowcast}) based on nowcast.} The last column shows the result for the out-of-sample expanding-window recursively estimated historical correction factor, after an initial warm-up period and excluding the COVID period from estimation. The minimum value (based on full precision) that improves on the mean forecast in each row (across the fixed factors) is bold.\\
Note: $^*$ indicates that the COVID period (2020Q1--2022Q4) was excluded from evaluation.}
\label{tab:CorrMeanF-realtime}
\end{table}

\begin{table}[H]
\centering \footnotesize

\singlespacing
\begin{tabular}{lcccccccccccc}
\hline
 & \multicolumn{10}{c}{Fixed correction factor} & recurs. \\
 \cline{2-11}
Period  & 0.1 & 0.2 & 0.3 & 0.4 & 0.5 & 0.6 & 0.7 & 0.8 & 0.9 & 1 & hist. \\
\hline
\multicolumn{1}{c}{UNEMP}\\
1969Q1--2025Q2  & 1.02 & 1.05 & 1.08 & 1.13 & 1.18 & 1.24 & 1.30 & 1.37 & 1.44 & 1.51 & 1.04 \\
1969Q1--2025Q2$^*$  & \textbf{0.99} & 0.99 & 1.00 & 1.03 & 1.06 & 1.10 & 1.14 & 1.20 & 1.26 & 1.32 & 1.07 \\
1969Q1--1999Q4  & \textbf{1.00} & 1.01 & 1.02 & 1.05 & 1.09 & 1.13 & 1.18 & 1.24 & 1.30 & 1.37 & 1.10 \\
2000Q1--2025Q2  & 1.02 & 1.05 & 1.09 & 1.14 & 1.19 & 1.25 & 1.32 & 1.39 & 1.46 & 1.53 & 1.03 \\
2000Q1--2025Q2$^*$  & 0.97 & 0.96 & \textbf{0.95} & 0.96 & 0.98 & 1.01 & 1.04 & 1.09 & 1.15 & 1.21 & 0.98 \\
2000Q1--2008Q4  & 0.99 & 0.99 & \textbf{0.98} & 0.99 & 1.00 & 1.01 & 1.03 & 1.05 & 1.08 & 1.11 & 0.99 \\
2000Q1--2019Q4  & 0.97 & 0.96 & \textbf{0.95} & 0.96 & 0.97 & 0.99 & 1.03 & 1.07 & 1.12 & 1.18 & 0.98 \\
2010Q1--2019Q4  & 1.00 & 1.02 & 1.04 & 1.07 & 1.11 & 1.16 & 1.21 & 1.27 & 1.33 & 1.40 & 1.01 \\
2022Q1--2025Q2  & 0.98 & \textbf{0.98} & 0.99 & 1.02 & 1.08 & 1.14 & 1.22 & 1.30 & 1.40 & 1.50 & 0.97 \\
\hline
\multicolumn{1}{c}{RGDP}\\
1969Q1--2025Q2     & 1.01 & 1.02 & 1.04 & 1.08 & 1.12 & 1.17 & 1.22 & 1.28 & 1.35 & 1.41 & 1.00 \\
1969Q1--2025Q2$^*$ & 1.00 & 1.01 & 1.03 & 1.06 & 1.10 & 1.15 & 1.20 & 1.26 & 1.32 & 1.39 & 1.00 \\
1969Q1--1999Q4     & 1.00 & 1.01 & 1.03 & 1.06 & 1.09 & 1.13 & 1.18 & 1.23 & 1.29 & 1.35 & 1.01 \\
2000Q1--2025Q2     & 1.01 & 1.02 & 1.05 & 1.08 & 1.13 & 1.18 & 1.24 & 1.30 & 1.37 & 1.44 & 1.00 \\
2000Q1--2025Q2$^*$ & 1.00 & 1.01 & 1.04 & 1.07 & 1.11 & 1.16 & 1.21 & 1.28 & 1.34 & 1.41 & 1.00 \\
2000Q1--2008Q4     & \textbf{0.99} & 1.01 & 1.05 & 1.11 & 1.19 & 1.29 & 1.39 & 1.51 & 1.63 & 1.75 & 1.00 \\
2000Q1--2019Q4     & 1.01 & 1.02 & 1.05 & 1.09 & 1.13 & 1.18 & 1.25 & 1.31 & 1.38 & 1.46 & 1.00 \\
2010Q1--2019Q4     & 1.00 & 1.02 & 1.05 & 1.09 & 1.14 & 1.20 & 1.26 & 1.34 & 1.41 & 1.50 & 1.00 \\
2022Q1--2025Q2     & \textbf{0.99} & 1.00 & 1.01 & 1.04 & 1.07 & 1.11 & 1.16 & 1.22 & 1.28 & 1.34 & 1.00 \\
\hline
\multicolumn{1}{c}{INDPROD}\\
1969Q1--2025Q2     & 1.00 & 1.01 & 1.03 & 1.06 & 1.10 & 1.14 & 1.19 & 1.25 & 1.31 & 1.38 & 1.05 \\
1969Q1--2025Q2$^*$ & \textbf{1.00} & 1.01 & 1.03 & 1.06 & 1.09 & 1.14 & 1.19 & 1.25 & 1.31 & 1.38 & 1.05 \\
1969Q1--1999Q4     & \textbf{1.00} & 1.01 & 1.03 & 1.06 & 1.09 & 1.14 & 1.19 & 1.24 & 1.31 & 1.37 & 1.05 \\
2000Q1--2025Q2     & 1.00 & 1.02 & 1.04 & 1.07 & 1.11 & 1.16 & 1.22 & 1.28 & 1.34 & 1.41 & 1.00 \\
2000Q1--2025Q2$^*$ & 1.00 & 1.01 & 1.03 & 1.06 & 1.10 & 1.15 & 1.20 & 1.26 & 1.32 & 1.39 & 1.00 \\
2000Q1--2008Q4     & 1.00 & 1.01 & 1.03 & 1.06 & 1.10 & 1.14 & 1.19 & 1.25 & 1.31 & 1.38 & 1.00 \\
2000Q1--2019Q4     & 1.00 & 1.01 & 1.03 & 1.06 & 1.10 & 1.14 & 1.19 & 1.25 & 1.32 & 1.38 & 1.00 \\
2010Q1--2019Q4     & 1.01 & 1.03 & 1.06 & 1.10 & 1.15 & 1.20 & 1.26 & 1.32 & 1.39 & 1.46 & 1.01 \\
2022Q1--2025Q2     & 1.09 & 1.21 & 1.36 & 1.53 & 1.71 & 1.91 & 2.10 & 2.31 & 2.51 & 2.72 & 1.04 \\
\hline
\multicolumn{1}{c}{CPI}\\
1981Q3--2025Q2     & 1.00 & 1.02 & 1.04 & 1.07 & 1.11 & 1.16 & 1.21 & 1.27 & 1.34 & 1.41 & 1.03 \\
1981Q3--2025Q2$^*$ & 1.01 & 1.03 & 1.06 & 1.10 & 1.15 & 1.20 & 1.27 & 1.33 & 1.41 & 1.48 & 1.03 \\
1981Q3--1999Q4     & 1.01 & 1.04 & 1.07 & 1.11 & 1.16 & 1.22 & 1.29 & 1.36 & 1.43 & 1.51 & 1.08 \\
2000Q1--2025Q2     & 1.00 & 1.01 & 1.03 & 1.06 & 1.10 & 1.15 & 1.20 & 1.25 & 1.32 & 1.38 & 1.02 \\
2000Q1--2025Q2$^*$ & 1.01 & 1.03 & 1.06 & 1.10 & 1.15 & 1.20 & 1.26 & 1.33 & 1.40 & 1.47 & 1.02 \\
2000Q1--2008Q4     & 1.02 & 1.04 & 1.06 & 1.09 & 1.12 & 1.15 & 1.19 & 1.23 & 1.27 & 1.31 & 0.99 \\
2000Q1--2019Q4     & 1.02 & 1.06 & 1.10 & 1.15 & 1.20 & 1.26 & 1.33 & 1.39 & 1.47 & 1.54 & 0.99 \\
2010Q1--2019Q4     & 1.02 & 1.05 & 1.09 & 1.13 & 1.18 & 1.24 & 1.31 & 1.38 & 1.45 & 1.53 & 0.99 \\
2022Q1--2025Q2     & 0.95 & 0.91 & 0.89 & \textbf{0.88} & 0.89 & 0.92 & 0.96 & 1.01 & 1.07 & 1.15 & 1.10 \\
\hline
\end{tabular}
\caption{
Relative RMSFE for UNEMP, RGDP, INDPROD, and CPI across different periods at forecast horizon $h=1$, using \textbf{the extra lagged error (\ref{eq:err-extra-lag}).} The last column shows the result for the out-of-sample expanding-window recursively estimated historical correction factor, after an initial warm-up period and excluding the COVID period from estimation. The minimum value (based on full precision) that improves on the mean forecast in each row (across the fixed factors) is bold.\\
Note: $^*$ indicates that the COVID period (2020Q1--2022Q4) was excluded from evaluation.}
\label{tab:CorrMeanF-realtime-v04a}
\end{table}

\subsection{Corrected Optimal Forecast: {An Illustration}}

We now turn to the correction of the optimal forecast. The main difficulty here is that the SPF is highly unbalanced, with many forecasters skipping periods, new forecasters joining the panel, and old ones dropping out, making the computation of the forecast-error covariance matrix $\Sigma$ for the optimal weights~(\ref{eq:BGweights}) challenging.
For this reason, we limit our demonstration to UNEMP for the period 2000Q1--2019Q4, before COVID. We further select 6 forecasters with at least 70 observations (out of 80). Even with this restriction, there are only a few time periods in which all of them submit their forecasts. We therefore need to implement a simple imputation using their previously submitted forecasts for the missing periods.
Other forms of imputation that prevent data leakage (i.e., using future information to impute the past) are also possible, see \citet{ahn2022comparison}, but since a relatively small number of observations are missing, the final results should not be sensitive. The precision of the individual forecasts is reported in the top section of Table~\ref{tbl:opt_weights_correction}.

\begin{table}[hbt]
\centering
\begingroup
\singlespacing
\small

\begin{tabular}{lcc}
\hline
{Forecasts} & {MSFE} & {relative MSFE} \\
\hline
\multicolumn{3}{l}{Individual forecasts} \\
ID 421        & 0.1557 & 0.9956 \\
ID 426        & 0.1736 & 1.1104 \\
ID 484        & 0.3679 & 2.3534 \\
ID 504        & 0.1681 & 1.0750 \\
ID 510        & 0.1948 & 1.2460 \\
ID 518        & 0.2255 & 1.4422 \\
\hline
\multicolumn{3}{l}{Combinations} \\
Mean Forecast       & 0.1563 & 1.0000 \\
-- corrected with $\gamma_i=0.5$ & 0.0802  &  0.5132 \\
-- corrected with $\gamma_i=0.65$ (fixed opt.) & 0.0729 & 0.4663 \\
-- corrected with hist. opt. factor &   0.0878   & 0.5613\\
Optimal Combination (via restr. OLS)        & 0.1613 & 1.0315 \\
-- corrected with $\gamma_i= 0.5$ & 0.0915 & 0.5850 \\
-- corrected with $\gamma_i=0.7$ (fixed opt.) & 0.0860 & 0.5502 \\
-- corrected with hist. opt. factor & 0.0914 & 0.5846 \\
One-step combination and correction                 & 0.0825 & 0.5275 \\
\multicolumn{3}{l}{(via restricted GLS)}  \\
\hline
\multicolumn{3}{l}{{Additional benchmarks}} \\ 
\citet{CoulsonRobins1993}, method (4) & 0.1682 & 1.0759 \\ 
\citet{CoulsonRobins1993}, method (5) & 0.1445 & 0.9240 \\ 
Correct-then-combine (\citealp{clements2022}) & 0.1894 & 1.2115 \\
\hline
\end{tabular}
\endgroup
\caption{Mean Squared Forecasting Errors (MSFE) and relative (to the mean forecast) MSFE for individual forecasts and combinations.
The evaluation begins in 2006, after the initial estimation period required by the recursively estimated combination and benchmark methods.} \label{tbl:opt_weights_correction}
\end{table}

The middle section of Table~\ref{tbl:opt_weights_correction} reports the precision of the combination methods. The simple mean of the forecasts outperforms each individual forecast (except one), consistent with the vast empirical literature; see, for example, \citet{WANG20231518}. The errors of the mean forecast exhibit high autocorrelation, as shown in Figure~\ref{fig:autocorr_mean}. Applying a correction factor $\gamma_i=0.5$ yields close to a 50\% improvement, which increases by another 5\% when using the optimal fixed correction factor $\gamma_i = 0.65$. Estimating the historically optimal correction factor slightly deteriorates forecasting performance but still provides a sizable improvement of close to 45\% over the original mean forecast. As expected, the correction removes the autocorrelation in the forecasting errors, as shown in Figure~\ref{fig:autocorr_mean_corrected}.

When running the restricted regression (\ref{eq:OLS}) of \citet{granger1984improved} over the full period 2000Q1--2019Q4, the optimal (in-sample) weights improve upon equal weights by construction. However, a more realistic approach is to use only the available observations $[t_0,t]$ to estimate the optimal (out-of-sample) weights and then use them to forecast $t+1$.
The performance of this approach corresponds to the optimal combination in Table~\ref{tbl:opt_weights_correction}. The optimal combination is slightly worse than the mean forecast, which is consistent with the vast empirical literature on the forecast combination puzzle, see \citet{claeskens2016forecast}.
The optimal forecast can be corrected using the default fixed factor 0.5, the optimal fixed factor 0.7, or a historically optimal factor $-1<\widehat{\gamma}_{i,t}<1$. All options outperform the mean forecast, which partially addresses the forecast combination puzzle, i.e., the corrected optimal forecast now performs better than the mean forecast, as expected by theory but often violated in practice.
The improvements are dramatic, with the default $\gamma_i=0.5$ correction performing more than 40\% better than the mean. The optimal fixed correction factor $\gamma_i=0.7$ further improves performance by another 3.5\%. Estimating the historical correction factor diminishes the improvement, yielding results similar to the default correction with $\gamma_i=0.5$.

The last row of the middle section in Table~\ref{tbl:opt_weights_correction} reports the case in which the estimation of the optimal weights and the correction factor is done simultaneously via GLS, as given by equation~(\ref{eq:GLScorrected}) in Section~\ref{sec:GLS}. This one-step procedure yields better results than the two-step forecasts analyzed in the previous paragraph because it is more efficient.
Its performance is almost catching up to the corrected mean forecast using the default correction factor of 0.5. Our practical recommendation is to use the corrected mean forecast until further research reveals insights into the modified forecasting puzzle: the corrected optimal forecast is better than the mean forecast, yet it is still outperformed by the corrected mean forecast. However, the GLS forecast is now a formidable contender and might outperform the corrected mean forecast in some applications.

{The bottom section of Table~\ref{tbl:opt_weights_correction} reports three additional benchmarks. The optimal combination estimated by restricted OLS in the middle section is the closest counterpart in our analysis to the static, no-dynamics method~(1) of \citet{CoulsonRobins1993}, although our specification omits the intercept and restricts the weights to sum to one. Relative to this method~(1)-type benchmark, method~(4), which includes the lagged dependent variable, increases the MSFE from 0.1613 to 0.1682, a deterioration of approximately 4.3\%. By contrast, the error-correction method~(5) reduces the MSFE to 0.1445, an improvement of approximately 10.4\% over the optimal-weight benchmark and 7.6\% over the mean forecast. The resulting ranking, with method~(5), followed by the method~(1)-type static combination and then method~(4), is consistent with the results reported by \citet{CoulsonRobins1993} for the Barnard data, for which the corresponding MSFEs are 15.83, 16.15, and 16.18. Finally, the correct-then-combine approach based on \citet{clements2022} produces a relative MSFE of 1.2115 and therefore performs 21.2\% worse than the mean forecast. Thus, in this application, correcting the individual forecasts before averaging does not outperform correcting the combined forecast and provides no evidence that the latter correction is superfluous. This finding is consistent with the greater estimation burden involved in recursively correcting each individual forecast.}

\begin{figure}[htbp]
    \centering

     \begin{subfigure}{0.45\textwidth}
        \centering
        \includegraphics[width=\linewidth]{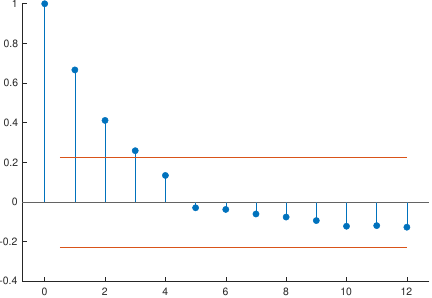}
        \caption{Mean forecast errors}
        \label{fig:autocorr_mean}
    \end{subfigure}
    \hfill
    \begin{subfigure}{0.45\textwidth}
        \centering
        \includegraphics[width=\linewidth]{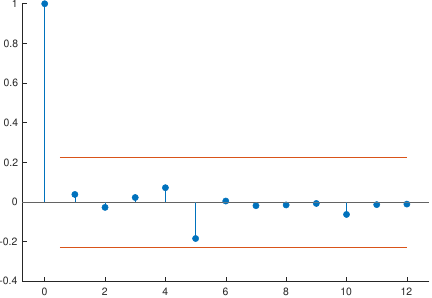}
        \caption{Corrected mean forecast errors (correction factor 0.65)}
        \label{fig:autocorr_mean_corrected}
    \end{subfigure}

    \vspace{0.5cm}

    \begin{subfigure}{0.45\textwidth}
        \centering
        \includegraphics[width=\linewidth]{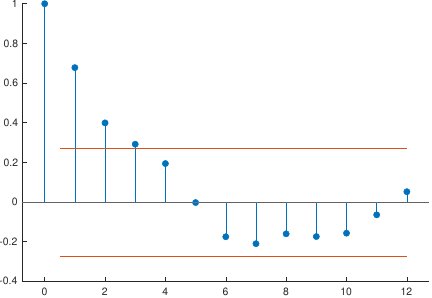}
        \caption{OLS regression forecast errors}
        \label{fig:autocorr_regr}
    \end{subfigure}
    \hfill
    \begin{subfigure}{0.45\textwidth}
        \centering
        \includegraphics[width=\linewidth]{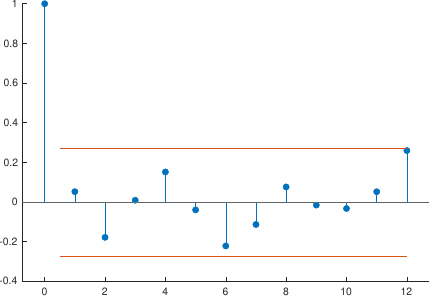}
        \caption{Corrected OLS regression forecast errors with historically optimal correction factor}
        \label{fig:autocorr_regr_corrected}
    \end{subfigure}

        \vspace{0.5cm}

    \begin{subfigure}{0.45\textwidth}
        \centering
        \includegraphics[width=\linewidth]{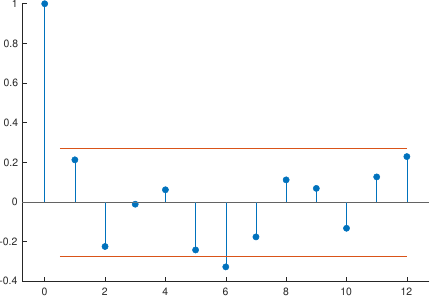}
        \caption{GLS regression forecast errors}
        \label{fig:autocorr_GLS}
    \end{subfigure}

    \caption{Autocorrelation of the forecast errors.}
    \label{fig:autocorr_regr_all}
\end{figure}

Figure~\ref{fig:autocorr_regr_all} investigates the autocorrelation of the regression forecasting errors. The OLS regression, which is equivalent to the optimal weights of \citet{bates1969combination}, produces forecasting errors that exhibit significant autocorrelation at multiple lags, as shown in Figure~\ref{fig:autocorr_regr}. Correcting the OLS forecast removes this autocorrelation, as shown in Figure~\ref{fig:autocorr_regr_corrected}. The same result is obtained by estimating the weight and correction simultaneously via GLS, as shown in Figure~\ref{fig:autocorr_GLS}.
{Figure~\ref{fig:autocorr-new-methods} in Appendix~\ref{app:empirical} presents the sample autocorrelation functions of the forecast errors generated by the three additional benchmarks. None of these methods successfully removes the serial dependence in the forecasting errors: appreciable autocorrelation remains after applying both dynamic methods of \citet{CoulsonRobins1993} and the individual efficiency corrections based on \citet{clements2022}, though the latter is successful in dealing with autocorrelation of order 2 and higher.}

\section{Conclusion} \label{sec:conclusion}

This paper shows that an important and largely overlooked source of forecast improvement lies not in refining combination weights, but in correcting combined forecasts when their errors are serially dependent. Building on the classic \citet{bates1969combination} example, we demonstrate that combined forecast errors often exhibit substantial autocorrelation and that exploiting this predictability can deliver gains that exceed those achieved by forecast combination itself. We formalize this idea within the conditional-risk framework of \citet{COW2024} and link forecast correction to efficient estimation of combination weights under time-series dependence using GLS. Empirical evidence from the U.S. Survey of Professional Forecasters across multiple variables and subsamples strongly supports the effectiveness of this approach.

A central empirical finding is the robustness of a simple, easily implementable correction: adding a fraction of the previous combined forecast error to the next-period combined forecast. Across most variables and subsamples—excluding periods dominated by extreme structural disruptions such as COVID—a correction factor near 0.5 consistently yields sizable improvements in forecast accuracy. The gains from this parsimonious adjustment are comparable to those from estimating historically optimal correction factors, while avoiding additional estimation uncertainty. As a result, the corrected mean forecast is a practical and robust default choice for applied forecasting.

The correction perspective also offers new insight into the forecast combination puzzle. Although optimal-weight combinations estimated via OLS often perform poorly out of sample, correcting these forecasts substantially improves their accuracy and frequently restores their theoretical advantage over simple averages. Jointly estimating weights and the correction factor via GLS further improves efficiency and eliminates residual autocorrelation in forecast errors, although in our application, the corrected mean forecast remains highly competitive.

These findings point to a promising and largely unexplored direction for future research. We encourage researchers to give \citet{Diebold1988}'s idea proper recognition by paying explicit attention to the time-series properties of combined forecast errors and to treat forecast correction as a complementary dimension to forecast combination. Rather than replacing existing techniques, correction can be naturally integrated into the multitude of established combination methods, as well as newer approaches, including machine-learning-based forecast ensembles. Incorporating simple correction mechanisms into both classical and modern forecasting frameworks has the potential to deliver substantial accuracy gains at minimal additional cost, and we view this as a fertile area for further theoretical and empirical work.

\section*{Code availability}

The code for the motivating example and the empirical investigation is available at \url{https://github.com/a-vasnev/GLS}.

{
\bibliographystyle{chicago}
\bibliography{COW_bib}
}

\newpage
\appendix

\titleformat{\section}
  {\normalfont\Large\bfseries}
  {Appendix~\thesection:}{0.5em}{}

\section{Proofs}\label{app:proofs}

\subsection*{Theorem \ref{Th1_combined}}
{
\begin{proof}
The finite-second-moment assumptions imply that
\[
\begin{aligned}
\E\!\left[
  \left(e_{T+h\mid T}^{\text{ZZ}}\right)^2
\right]
&=
\E\!\left[
  \left(
    y_{T+h}-f_{T+h\mid T}^{\text{ZZ}}
  \right)^2
\right]\\
&\leq
2\E(y_{T+h}^2)
+
2\E\!\left[
  \left(f_{T+h\mid T}^{\text{ZZ}}\right)^2
\right]
<\infty.
\end{aligned}
\]
Furthermore, conditional Jensen's inequality gives
\[
\begin{aligned}
\E\!\left[
  \left(b_T^{\text{ZZ}}\right)^2
\right]
&=
\E\!\left\{
  \left[
    \E\!\left(
      e_{T+h\mid T}^{\text{ZZ}}
      \mid I_T
    \right)
  \right]^2
\right\}\\
&\leq
\E\!\left[
  \left(e_{T+h\mid T}^{\text{ZZ}}\right)^2
\right]
<\infty.
\end{aligned}
\]
Thus, the oracle correction term $b_T^{\text{ZZ}}$ is also
square-integrable.

Because $f_{T+h\mid T}^{\text{ZZ}}$ is $I_T$-measurable,
\[
\begin{aligned}
\E(y_{T+h}\mid I_T)
&=
\E\!\left(
  f_{T+h\mid T}^{\text{ZZ}}
  +e_{T+h\mid T}^{\text{ZZ}}
  \mid I_T
\right)\\
&=
f_{T+h\mid T}^{\text{ZZ}}
+
\E\!\left(
  e_{T+h\mid T}^{\text{ZZ}}
  \mid I_T
\right)\\
&=
f_{T+h\mid T}^{\text{ZZ}}
+b_T^{\text{ZZ}}
=
f_{T+h\mid T}^{\text{CZZ}}.
\end{aligned}
\]
Hence, the oracle-corrected forecast is the conditional expectation of
$y_{T+h}$ given $I_T$.

The corresponding corrected forecast error is
\[
\begin{aligned}
e_{T+h\mid T}^{\text{CZZ}}
&=
y_{T+h}-f_{T+h\mid T}^{\text{CZZ}}\\
&=
y_{T+h}
-f_{T+h\mid T}^{\text{ZZ}}
-b_T^{\text{ZZ}}\\
&=
e_{T+h\mid T}^{\text{ZZ}}
-b_T^{\text{ZZ}}.
\end{aligned}
\]
It has zero conditional mean because
\[
\begin{aligned}
\E\!\left(
  e_{T+h\mid T}^{\text{CZZ}}
  \mid I_T
\right)
&=
\E\!\left(
  e_{T+h\mid T}^{\text{ZZ}}
  -b_T^{\text{ZZ}}
  \mid I_T
\right)\\
&=
\E\!\left(
  e_{T+h\mid T}^{\text{ZZ}}
  \mid I_T
\right)
-b_T^{\text{ZZ}}\\
&=0.
\end{aligned}
\]

Now let $\widetilde b_T^{\text{ZZ}}$ be any square-integrable,
$I_T$-measurable additive correction term. The forecast error obtained
using this correction can be written as
\[
\begin{aligned}
e_{T+h\mid T}^{\text{ZZ}}
-\widetilde b_T^{\text{ZZ}}
&=
e_{T+h\mid T}^{\text{CZZ}}
+b_T^{\text{ZZ}}
-\widetilde b_T^{\text{ZZ}}.
\end{aligned}
\]
Therefore,
\[
\begin{aligned}
&\E\!\left[
  \left(
    e_{T+h\mid T}^{\text{ZZ}}
    -\widetilde b_T^{\text{ZZ}}
  \right)^2
  \mid I_T
\right]\\
&\quad =
\E\!\left[
  \left(
    e_{T+h\mid T}^{\text{CZZ}}
    +b_T^{\text{ZZ}}
    -\widetilde b_T^{\text{ZZ}}
  \right)^2
  \mid I_T
\right]\\
&\quad =
\E\!\left[
  \left(e_{T+h\mid T}^{\text{CZZ}}\right)^2
  \mid I_T
\right]\\
&\qquad\quad
+
2\left(
  b_T^{\text{ZZ}}
  -\widetilde b_T^{\text{ZZ}}
\right)
\E\!\left(
  e_{T+h\mid T}^{\text{CZZ}}
  \mid I_T
\right)
+
\left(
  b_T^{\text{ZZ}}
  -\widetilde b_T^{\text{ZZ}}
\right)^2\\
&\quad =
\E\!\left[
  \left(e_{T+h\mid T}^{\text{CZZ}}\right)^2
  \mid I_T
\right]
+
\left(
  b_T^{\text{ZZ}}
  -\widetilde b_T^{\text{ZZ}}
\right)^2.
\end{aligned}
\]
The final term is nonnegative. Therefore,
$b_T^{\text{ZZ}}$ minimizes conditional MSE among all
square-integrable, $I_T$-measurable additive corrections. The
minimizer is unique up to almost-sure equality.

Taking unconditional expectations gives
\[
\begin{aligned}
&\E\!\left[
  \left(
    e_{T+h\mid T}^{\text{ZZ}}
    -\widetilde b_T^{\text{ZZ}}
  \right)^2
\right]\\
&\quad =
\E\!\left[
  \left(e_{T+h\mid T}^{\text{CZZ}}\right)^2
\right]
+
\E\!\left[
  \left(
    b_T^{\text{ZZ}}
    -\widetilde b_T^{\text{ZZ}}
  \right)^2
\right].
\end{aligned}
\]
It follows that $b_T^{\text{ZZ}}$ also minimizes unconditional MSE,
again uniquely up to almost-sure equality.

In particular, setting $\widetilde b_T^{\text{ZZ}}=0$ gives
\[
\begin{aligned}
&\E\!\left[
  \left(e_{T+h\mid T}^{\text{ZZ}}\right)^2
  \mid I_T
\right]
-
\E\!\left[
  \left(e_{T+h\mid T}^{\text{CZZ}}\right)^2
  \mid I_T
\right]\\
&\qquad =
\left(b_T^{\text{ZZ}}\right)^2
=
\left[
  \E\!\left(
    e_{T+h\mid T}^{\text{ZZ}}
    \mid I_T
  \right)
\right]^2
\geq 0
\qquad\text{a.s.}
\end{aligned}
\]
Taking expectations once more yields
\[
\begin{aligned}
&\E\!\left[
  \left(e_{T+h\mid T}^{\text{ZZ}}\right)^2
\right]
-
\E\!\left[
  \left(e_{T+h\mid T}^{\text{CZZ}}\right)^2
\right]\\
&\qquad =
\E\!\left[
  \left(b_T^{\text{ZZ}}\right)^2
\right]
=
\E\!\left\{
  \left[
    \E\!\left(
      e_{T+h\mid T}^{\text{ZZ}}
      \mid I_T
    \right)
  \right]^2
\right\}
\geq 0.
\end{aligned}
\]
The conditional improvement is strict precisely on the event
\[
\left\{b_T^{\text{ZZ}}\neq0\right\}
=
\left\{
  \E\!\left(
    e_{T+h\mid T}^{\text{ZZ}}
    \mid I_T
  \right)\neq0
\right\},
\]
and the unconditional improvement is strict if and only if
\[
\Pr\!\left(b_T^{\text{ZZ}}\neq0\right)>0.
\]

Finally, because
\[
f_{T+h\mid T}^{\text{CZZ}}
=
\E(y_{T+h}\mid I_T),
\]
the oracle-corrected forecast does not depend on the particular
baseline forecast indexed by $\text{ZZ}$. Thus, all oracle-corrected
forecasts based on the same information set $I_T$ coincide.
\end{proof}
}

\subsection*{Theorem \ref{thm:oracle_corrected_combination}}
{
\begin{proof}
The assumed error representation gives
\[
y_{T+h}
=
f^{\mathrm{C}0}_{T+h\mid T}
+
\xi_{T+h}.
\]
Because \(f^{\mathrm{C}0}_{T+h\mid T}\) is \(I_T\)-measurable and
\(\E(\xi_{T+h}\mid I_T)=0\),
\[
f^{\mathrm{C}0}_{T+h\mid T}
=
\E(y_{T+h}\mid I_T).
\]
For any \(I_T\)-measurable forecast \(g_T\),
\[
\begin{aligned}
\E[(y_{T+h}-g_T)^2\mid I_T]
&=
\E\!\left[
  \left(
    \xi_{T+h}
    +f^{\mathrm{C}0}_{T+h\mid T}-g_T
  \right)^2
  \mid I_T
\right] \\
&=
\E(\xi_{T+h}^2\mid I_T)
+
\left(
  f^{\mathrm{C}0}_{T+h\mid T}-g_T
\right)^2.
\end{aligned}
\]
The cross term vanishes because
\(\E(\xi_{T+h}\mid I_T)=0\). Since
\[
\E\!\left[
  \left(y_{T+h}-f^{\mathrm{C}0}_{T+h\mid T}\right)^2
  \mid I_T
\right]
=
\E(\xi_{T+h}^2\mid I_T),
\]
setting \(g_T=f_{c,T+h\mid T}(\vw)=\vw'\vf_{T+h\mid T}\) yields
\[
\begin{aligned}
&\E\!\left[
  \left(y_{T+h}-f_{c,T+h\mid T}(\vw)\right)^2
  \mid I_T
\right]
-
\E\!\left[
  \left(y_{T+h}-f^{\mathrm{C}0}_{T+h\mid T}\right)^2
  \mid I_T
\right]\\
&\qquad=
\left(
  f^{\mathrm{C}0}_{T+h\mid T}
  -f_{c,T+h\mid T}(\vw)
\right)^2
\geq 0.
\end{aligned}
\]
Taking expectations gives the unconditional MSE inequality.
Finally, when \(\vw=\vw_0\),
\[
f^{\mathrm{C}0}_{T+h\mid T}
-f_{c,T+h\mid T}(\vw_0)
=
\gamma_0 e_{c,T\mid T-h}(\vw_0).
\]
Thus, the conditional MSE reduction is
\(\gamma_0^2 e_{c,T\mid T-h}(\vw_0)^2\), and taking expectations
gives the stated unconditional reduction.
\end{proof}
}

\subsection*{Theorem \ref{thm:oracle_corrected_combinations}}
{
\begin{proof}
Let $f_J\in\mathcal C(J_T)$. Because
$\mathcal C(J_T)\subseteq\mathcal C(I_T)$, the forecast $f_J$ is
also admissible under $I_T$. Therefore, by the definition of
$f^*(I_T)$,
\[
\E\!\left[
  L\!\left(y_{T+h}-f^*(I_T)\right)
  \mid I_T
\right]
\leq
\E\!\left[
  L(y_{T+h}-f_J)
  \mid I_T
\right].
\]
Taking conditional expectations given $J_T$ and applying the law of
iterated expectations gives
\[
\E\!\left[
  L\!\left(y_{T+h}-f^*(I_T)\right)
  \mid J_T
\right]
\leq
\E\!\left[
  L(y_{T+h}-f_J)
  \mid J_T
\right].
\]
Setting $f_J=f^*(J_T)$ proves the result.
\end{proof}
}

\subsection*{Theorem \ref{Th4}}
\begin{proof}
Under the restriction $\vw'\viota=1$, we have $\ve_{c} = \vy - \mathbf{F}\vw = \vy\viota'\vw - \mathbf{F}\vw = (\vy\viota'-\mathbf{F})\vw = \mathbf E\vw$. Therefore, the theoretical and empirical risks can be written as
\[
R(\vw)=\vw'\mathbf E'\Omega^{-1}(\vgamma)\mathbf E\vw\;\;\text{and}\;\;
R_T(\vw)=\vw'\mathbf E'\Omega^{-1}(\widehat{\vgamma}^{\text{GLS}})\mathbf E\vw,
\]
respectively. { Hence,
\[
|R_T(\vw)-R(\vw)|
\leq \|\mathbf E'(\Omega^{-1}(\widehat{\vgamma}^{\text{GLS}})-\Omega^{-1}(\vgamma))\mathbf E\|_{\infty} \|\vw\|_1^{2}
=a_{T} \|\vw\|_1^{2}.
\]
Using the same argument as in the proof of Theorem~1 of \citet{fan2012vast}, the results follow from $\|\vw^\text{opt}\|_1\leq{c}$, $\|\widehat{\vw}^\text{GLS}\|_1\leq{c}$, and the optimality conditions $R(\vw^\text{opt})\leq R(\widehat{\vw}^\text{GLS})$ and $R_T(\widehat{\vw}^\text{GLS}) \leq R_T(\vw^\text{opt})$.}
\end{proof}

{
\subsection*{Combination of two AR(1) forecast errors}

\begin{theorem}
\label{thm:combined-error-arma21}
Let
$
    x_{i,t}=e_{i,t\mid t-1},
    \;
    (1-\rho_iB)x_{i,t}=\varepsilon_{i,t},
    \; i=1,2,
$
where $|\rho_i|<1$, $B$ is the backshift operator, and
$\varepsilon_{1,t}$ and $\varepsilon_{2,t}$ are mutually independent
Gaussian white-noise processes with variances $\sigma_1^2>0$ and
$\sigma_2^2>0$. For $w\in\mathbb{R}$, define the combined forecast error by
\[
    x_{c,t}=w x_{1,t}+(1-w)x_{2,t}.
\]
Then there exist a scalar Gaussian white-noise process $u_t$ and a
coefficient $\theta$ such that
\begin{equation}
    \left[
        1-(\rho_1+\rho_2)B+\rho_1\rho_2B^2
    \right]x_{c,t}
    =
    (1+\theta B)u_t.
    \label{eq:combined-error-arma21}
\end{equation}
Thus, the combined forecast error admits an ARMA$(2,1)$
representation. Except under special parameter restrictions that
reduce the model order ($\rho_1=\rho_2$, $w=0$, $w=1$, $\rho_1\rho_2=0$, or $\theta=0$), this representation is minimal, so the
combination is generally ARMA$(2,1)$ rather than AR$(1)$.
\end{theorem}

\begin{proof}
Applying the product of the two AR operators to the combined error
gives
\begin{align}
 &(1-\rho_1B)(1-\rho_2B)x_{c,t} \notag\\
 &\quad =
 w(1-\rho_2B)\varepsilon_{1,t}
 +(1-w)(1-\rho_1B)\varepsilon_{2,t} \notag\\
 &\quad =
 w\varepsilon_{1,t}+(1-w)\varepsilon_{2,t}
 -w\rho_2\varepsilon_{1,t-1}
 -(1-w)\rho_1\varepsilon_{2,t-1}
 \equiv v_t .
 \label{eq:combined-arma-derivation}
\end{align}
The process $v_t$ has autocovariances
\begin{align*}
    q_0
    &=
    w^2\sigma_1^2(1+\rho_2^2)
    +(1-w)^2\sigma_2^2(1+\rho_1^2),\\
    q_1
    &=
    -w^2\rho_2\sigma_1^2
    -(1-w)^2\rho_1\sigma_2^2,\\
    q_k&=0,\qquad |k|\geq 2.
\end{align*}
Moreover, $2|q_1|\leq q_0$. Letting $r=q_1/q_0$, define
\[
    \theta=
    \begin{cases}
    \displaystyle
    \frac{1-\sqrt{1-4r^2}}{2r},&r\neq0,\\[6pt]
    0,&r=0,
    \end{cases}
    \qquad
    \sigma_u^2=\frac{q_0}{1+\theta^2}.
\]
These definitions imply
\[
    q_0=\sigma_u^2(1+\theta^2),
    \qquad
    q_1=\sigma_u^2\theta.
\]
Because $v_t$ is Gaussian and its autocovariances vanish beyond lag
one, it therefore admits the scalar MA$(1)$ representation
\[
    v_t=u_t+\theta u_{t-1},
    \qquad
    \operatorname{Var}(u_t)=\sigma_u^2.
\]
Substitution into~\eqref{eq:combined-arma-derivation}, followed by
expansion of $(1-\rho_1B)(1-\rho_2B)$, establishes
equation~\eqref{eq:combined-error-arma21}. Cases such as
$\rho_1=\rho_2$, $w=0$, $w=1$, $\rho_1\rho_2=0$, or
$\theta=0$ may reduce the AR or MA order; away from such special
restrictions, the orders are two and one, respectively. \qed
\end{proof}
}

\newpage
{
\captionsetup{labelfont={bf}}
\section{Additional Empirical Results}\label{app:empirical}

\begin{table}[H]

\centering \footnotesize
\singlespacing
\begin{tabular}{lcccccccccccc}
\hline
 & \multicolumn{10}{c}{Fixed correction factor} & recurs. \\
\cline{2-11}
Period & 0.1 & 0.2 & 0.3 & 0.4 & 0.5 & 0.6 & 0.7 & 0.8 & 0.9 & 1 & hist. \\
\hline
\multicolumn{1}{c}{UNEMP}\\
1969Q1--2025Q2 & 1.02 & 1.05 & 1.09 & 1.13 & 1.18 & 1.24 & 1.30 & 1.37 & 1.44 & 1.52 & 1.13 \\
1969Q1--2025Q2$^*$ & 0.97 & 0.94 & 0.93 & \textbf{0.93} & 0.94 & 0.96 & 0.99 & 1.02 & 1.07 & 1.12 & 0.97 \\
1969Q1--1999Q4 & 0.97 & 0.96 & \textbf{0.95} & 0.96 & 0.98 & 1.00 & 1.04 & 1.08 & 1.13 & 1.19 & 1.01 \\
2000Q1--2025Q2 & 1.03 & 1.07 & 1.12 & 1.17 & 1.23 & 1.30 & 1.37 & 1.44 & 1.51 & 1.59 & 1.15 \\
2000Q1--2025Q2$^*$ & 0.95 & 0.91 & 0.88 & 0.85 & \textbf{0.85} & 0.85 & 0.86 & 0.89 & 0.93 & 0.98 & 0.88 \\
2000Q1--2008Q4 & 0.97 & 0.95 & 0.94 & 0.93 & \textbf{0.92} & 0.92 & 0.93 & 0.94 & 0.96 & 0.98 & 0.94 \\
2000Q1--2019Q4 & 0.95 & 0.91 & 0.88 & 0.86 & \textbf{0.85} & 0.85 & 0.86 & 0.89 & 0.92 & 0.97 & 0.88 \\
2010Q1--2019Q4 & 0.97 & 0.95 & \textbf{0.94} & 0.95 & 0.97 & 1.00 & 1.04 & 1.09 & 1.15 & 1.22 & 0.95 \\
2022Q1--2025Q2 & 0.93 & 0.88 & 0.85 & 0.82 & \textbf{0.82} & 0.84 & 0.87 & 0.92 & 0.98 & 1.06 & 0.83 \\
\hline
\multicolumn{1}{c}{RGDP}\\
1969Q1--2025Q2 & 0.97 & 0.96 & \textbf{0.95} & 0.96 & 0.97 & 1.00 & 1.03 & 1.08 & 1.13 & 1.18 & 0.98 \\
1969Q1--2025Q2$^*$ & 0.97 & 0.94 & 0.93 & \textbf{0.93} & 0.94 & 0.96 & 0.99 & 1.03 & 1.07 & 1.12 & 0.96 \\
1969Q1--1999Q4 & 0.97 & 0.95 & 0.94 & \textbf{0.94} & 0.95 & 0.97 & 1.00 & 1.03 & 1.07 & 1.12 & 1.02 \\
2000Q1--2025Q2 & 0.98 & 0.96 & \textbf{0.96} & 0.96 & 0.98 & 1.01 & 1.05 & 1.09 & 1.15 & 1.21 & 0.96 \\
2000Q1--2025Q2$^*$ & 0.97 & 0.94 & 0.93 & \textbf{0.92} & 0.93 & 0.95 & 0.98 & 1.02 & 1.07 & 1.13 & 0.93 \\
2000Q1--2008Q4 & 0.94 & 0.90 & 0.87 & \textbf{0.86} & 0.87 & 0.90 & 0.94 & 1.00 & 1.07 & 1.15 & 0.87 \\
2000Q1--2019Q4 & 0.97 & 0.95 & \textbf{0.95} & 0.95 & 0.97 & 1.00 & 1.03 & 1.08 & 1.13 & 1.19 & 0.96 \\
2010Q1--2019Q4 & 0.97 & \textbf{0.96} & 0.96 & 0.98 & 1.01 & 1.05 & 1.10 & 1.16 & 1.24 & 1.31 & 0.98 \\
2022Q1--2025Q2 & 0.95 & 0.92 & 0.89 & 0.88 & \textbf{0.87} & 0.88 & 0.90 & 0.92 & 0.96 & 1.01 & 0.89 \\
\hline
\multicolumn{1}{c}{INDPROD}\\
1969Q1--2025Q2 & 0.97 & 0.95 & 0.93 & \textbf{0.93} & 0.94 & 0.96 & 0.99 & 1.03 & 1.07 & 1.13 & 0.94 \\
1969Q1--2025Q2$^*$ & 0.97 & 0.94 & 0.93 & \textbf{0.93} & 0.94 & 0.96 & 0.98 & 1.02 & 1.07 & 1.12 & 0.94 \\
1969Q1--1999Q4 & 0.97 & 0.94 & 0.93 & \textbf{0.93} & 0.93 & 0.95 & 0.98 & 1.02 & 1.06 & 1.12 & 0.93 \\
2000Q1--2025Q2 & 0.97 & 0.96 & \textbf{0.95} & 0.96 & 0.97 & 1.00 & 1.03 & 1.07 & 1.12 & 1.18 & 0.96 \\
2000Q1--2025Q2$^*$ & 0.97 & 0.95 & \textbf{0.94} & 0.94 & 0.95 & 0.97 & 1.00 & 1.04 & 1.09 & 1.15 & 0.94 \\
2000Q1--2008Q4 & 0.97 & 0.95 & 0.93 & \textbf{0.93} & 0.94 & 0.96 & 0.99 & 1.03 & 1.07 & 1.13 & 0.94 \\
2000Q1--2019Q4 & 0.97 & 0.95 & 0.93 & \textbf{0.93} & 0.94 & 0.96 & 0.99 & 1.03 & 1.08 & 1.13 & 0.94 \\
2010Q1--2019Q4 & 0.98 & 0.97 & \textbf{0.96} & 0.97 & 0.99 & 1.02 & 1.06 & 1.11 & 1.16 & 1.22 & 0.97 \\
2022Q1--2025Q2 & 1.07 & 1.20 & 1.37 & 1.57 & 1.79 & 2.03 & 2.27 & 2.52 & 2.77 & 3.03 & 1.52 \\
\hline
\multicolumn{1}{c}{CPI}\\
1981Q3--2025Q2 & \textbf{0.99} & 1.00 & 1.01 & 1.04 & 1.07 & 1.11 & 1.16 & 1.22 & 1.28 & 1.35 & 1.02 \\
1981Q3--2025Q2$^*$ & \textbf{1.00} & 1.01 & 1.03 & 1.06 & 1.11 & 1.15 & 1.21 & 1.27 & 1.34 & 1.41 & 1.02 \\
1981Q3--1999Q4 & \textbf{0.99} & 1.00 & 1.01 & 1.04 & 1.07 & 1.12 & 1.17 & 1.23 & 1.30 & 1.37 & 1.02 \\
2000Q1--2025Q2 & \textbf{0.99} & 1.00 & 1.01 & 1.04 & 1.07 & 1.11 & 1.16 & 1.22 & 1.28 & 1.34 & 1.02 \\
2000Q1--2025Q2$^*$ & 1.00 & 1.02 & 1.04 & 1.07 & 1.12 & 1.17 & 1.22 & 1.29 & 1.35 & 1.43 & 1.02 \\
2000Q1--2008Q4 & 1.01 & 1.03 & 1.05 & 1.08 & 1.11 & 1.14 & 1.17 & 1.21 & 1.25 & 1.29 & 1.01 \\
2000Q1--2019Q4 & 1.02 & 1.05 & 1.09 & 1.14 & 1.19 & 1.25 & 1.32 & 1.39 & 1.46 & 1.53 & 1.01 \\
2010Q1--2019Q4 & 1.01 & 1.04 & 1.07 & 1.11 & 1.16 & 1.22 & 1.28 & 1.35 & 1.42 & 1.49 & 1.00 \\
2022Q1--2025Q2 & 0.93 & 0.88 & 0.84 & 0.82 & \textbf{0.81} & 0.82 & 0.84 & 0.89 & 0.94 & 1.01 & 1.06 \\
\hline
\end{tabular}
\caption{Relative RMSFE for UNEMP, RGDP, INDPROD, and CPI at forecast horizon $h=2$ (relative to the original mean forecast at the same horizon). Fixed factors, the recursive column, bolding, and starred evaluation samples are defined as in Table~\ref{tab:CorrMeanF}; correction factors are estimated separately for this horizon.}
\label{tab:CorrMeanF-h2}
\end{table}
\clearpage
\begin{table}[H]

\centering \footnotesize
\singlespacing
\begin{tabular}{lcccccccccccc}
\hline
 & \multicolumn{10}{c}{Fixed correction factor} & recurs. \\
\cline{2-11}
Period & 0.1 & 0.2 & 0.3 & 0.4 & 0.5 & 0.6 & 0.7 & 0.8 & 0.9 & 1 & hist. \\
\hline
\multicolumn{1}{c}{UNEMP}\\
1969Q1--2025Q2 & 1.01 & 1.03 & 1.06 & 1.10 & 1.15 & 1.20 & 1.26 & 1.32 & 1.39 & 1.46 & 1.06 \\
1969Q1--2025Q2$^*$ & 0.98 & 0.97 & \textbf{0.97} & 0.98 & 1.00 & 1.03 & 1.07 & 1.12 & 1.17 & 1.24 & 1.01 \\
1969Q1--1999Q4 & 0.99 & \textbf{0.99} & 1.00 & 1.02 & 1.04 & 1.08 & 1.13 & 1.18 & 1.23 & 1.30 & 1.05 \\
2000Q1--2025Q2 & 1.02 & 1.05 & 1.09 & 1.13 & 1.19 & 1.24 & 1.31 & 1.38 & 1.45 & 1.52 & 1.07 \\
2000Q1--2025Q2$^*$ & 0.96 & 0.92 & 0.90 & \textbf{0.90} & 0.90 & 0.92 & 0.96 & 1.00 & 1.05 & 1.12 & 0.93 \\
2000Q1--2008Q4 & 0.99 & 0.98 & \textbf{0.98} & 0.98 & 0.99 & 1.01 & 1.03 & 1.06 & 1.09 & 1.13 & 0.98 \\
2000Q1--2019Q4 & 0.96 & 0.93 & 0.92 & \textbf{0.91} & 0.91 & 0.93 & 0.95 & 0.99 & 1.03 & 1.08 & 0.95 \\
2010Q1--2019Q4 & \textbf{0.97} & 0.97 & 1.02 & 1.10 & 1.21 & 1.34 & 1.49 & 1.64 & 1.81 & 1.98 & 0.99 \\
2022Q1--2025Q2 & 0.87 & 0.77 & \textbf{0.71} & 0.72 & 0.77 & 0.87 & 1.01 & 1.16 & 1.33 & 1.50 & 0.74 \\
\hline
\multicolumn{1}{c}{RGDP}\\
1969Q1--2025Q2 & 0.98 & \textbf{0.98} & 0.98 & 0.99 & 1.02 & 1.05 & 1.09 & 1.14 & 1.19 & 1.25 & 0.99 \\
1969Q1--2025Q2$^*$ & 0.98 & 0.97 & \textbf{0.96} & 0.97 & 0.99 & 1.02 & 1.06 & 1.10 & 1.16 & 1.21 & 0.98 \\
1969Q1--1999Q4 & 0.98 & 0.97 & \textbf{0.96} & 0.97 & 0.98 & 1.01 & 1.04 & 1.08 & 1.13 & 1.18 & 1.01 \\
2000Q1--2025Q2 & 0.98 & \textbf{0.98} & 0.99 & 1.00 & 1.03 & 1.06 & 1.11 & 1.16 & 1.21 & 1.28 & 0.98 \\
2000Q1--2025Q2$^*$ & 0.98 & 0.97 & \textbf{0.97} & 0.98 & 1.00 & 1.03 & 1.07 & 1.11 & 1.17 & 1.23 & 0.97 \\
2000Q1--2008Q4 & 0.97 & 0.95 & \textbf{0.94} & 0.95 & 0.98 & 1.02 & 1.07 & 1.13 & 1.19 & 1.27 & 0.95 \\
2000Q1--2019Q4 & 0.99 & \textbf{0.98} & 0.99 & 1.01 & 1.03 & 1.07 & 1.12 & 1.17 & 1.23 & 1.29 & 0.99 \\
2010Q1--2019Q4 & 0.98 & \textbf{0.97} & 0.98 & 1.00 & 1.03 & 1.07 & 1.12 & 1.18 & 1.24 & 1.31 & 0.98 \\
2022Q1--2025Q2 & 0.97 & 0.94 & 0.93 & \textbf{0.93} & 0.94 & 0.96 & 0.98 & 1.02 & 1.07 & 1.12 & 0.94 \\
\hline
\multicolumn{1}{c}{INDPROD}\\
1969Q1--2025Q2 & 0.97 & 0.96 & \textbf{0.95} & 0.96 & 0.98 & 1.00 & 1.04 & 1.08 & 1.13 & 1.18 & 0.97 \\
1969Q1--2025Q2$^*$ & 0.97 & 0.96 & \textbf{0.95} & 0.96 & 0.97 & 1.00 & 1.03 & 1.08 & 1.13 & 1.18 & 0.96 \\
1969Q1--1999Q4 & 0.98 & 0.96 & \textbf{0.96} & 0.96 & 0.98 & 1.01 & 1.04 & 1.08 & 1.13 & 1.19 & 0.97 \\
2000Q1--2025Q2 & 0.97 & 0.95 & \textbf{0.94} & 0.94 & 0.96 & 0.98 & 1.01 & 1.05 & 1.10 & 1.16 & 0.95 \\
2000Q1--2025Q2$^*$ & 0.97 & 0.95 & 0.94 & \textbf{0.93} & 0.95 & 0.97 & 1.00 & 1.04 & 1.08 & 1.14 & 0.94 \\
2000Q1--2008Q4 & 0.97 & 0.95 & 0.94 & \textbf{0.94} & 0.95 & 0.97 & 1.00 & 1.04 & 1.08 & 1.14 & 0.94 \\
2000Q1--2019Q4 & 0.97 & 0.95 & 0.93 & \textbf{0.93} & 0.94 & 0.96 & 0.99 & 1.03 & 1.08 & 1.13 & 0.94 \\
2010Q1--2019Q4 & 0.97 & \textbf{0.96} & 0.96 & 0.98 & 1.02 & 1.07 & 1.13 & 1.20 & 1.28 & 1.37 & 0.96 \\
2022Q1--2025Q2 & 0.97 & \textbf{0.97} & 0.99 & 1.03 & 1.09 & 1.17 & 1.26 & 1.36 & 1.47 & 1.58 & 0.99 \\
\hline
\multicolumn{1}{c}{CPI}\\
1981Q3--2025Q2 & 0.99 & \textbf{0.98} & 0.99 & 1.01 & 1.04 & 1.07 & 1.12 & 1.17 & 1.23 & 1.29 & 1.01 \\
1981Q3--2025Q2$^*$ & 0.99 & \textbf{0.99} & 1.00 & 1.02 & 1.05 & 1.10 & 1.15 & 1.21 & 1.27 & 1.35 & 1.01 \\
1981Q3--1999Q4 & 0.97 & 0.95 & \textbf{0.94} & 0.94 & 0.96 & 0.98 & 1.02 & 1.07 & 1.12 & 1.19 & 0.94 \\
2000Q1--2025Q2 & \textbf{0.99} & 0.99 & 1.00 & 1.02 & 1.05 & 1.09 & 1.14 & 1.19 & 1.25 & 1.31 & 1.02 \\
2000Q1--2025Q2$^*$ & \textbf{0.99} & 1.00 & 1.02 & 1.04 & 1.08 & 1.13 & 1.19 & 1.25 & 1.32 & 1.39 & 1.03 \\
2000Q1--2008Q4 & 1.00 & \textbf{1.00} & 1.00 & 1.01 & 1.02 & 1.03 & 1.05 & 1.08 & 1.10 & 1.13 & 1.00 \\
2000Q1--2019Q4 & 1.01 & 1.04 & 1.07 & 1.11 & 1.16 & 1.22 & 1.28 & 1.34 & 1.41 & 1.48 & 1.05 \\
2010Q1--2019Q4 & 1.02 & 1.05 & 1.09 & 1.13 & 1.18 & 1.24 & 1.31 & 1.38 & 1.45 & 1.53 & 1.01 \\
2022Q1--2025Q2 & 0.93 & 0.86 & 0.82 & 0.80 & \textbf{0.80} & 0.82 & 0.86 & 0.92 & 0.99 & 1.08 & 0.97 \\
\hline
\end{tabular}
\caption{Relative RMSFE for UNEMP, RGDP, INDPROD, and CPI at forecast horizon $h=3$ (relative to the original mean forecast at the same horizon). Fixed factors, the recursive column, bolding, and starred evaluation samples are defined as in Table~\ref{tab:CorrMeanF}; correction factors are estimated separately for this horizon.}
\label{tab:CorrMeanF-h3}
\end{table}
\clearpage
\begin{table}[H]

\centering \footnotesize
\singlespacing
\begin{tabular}{lcccccccccccc}
\hline
 & \multicolumn{10}{c}{Fixed correction factor} & recurs. \\
\cline{2-11}
Period & 0.1 & 0.2 & 0.3 & 0.4 & 0.5 & 0.6 & 0.7 & 0.8 & 0.9 & 1 & hist. \\
\hline
\multicolumn{1}{c}{UNEMP}\\
1969Q1--2025Q2 & 1.01 & 1.03 & 1.06 & 1.10 & 1.15 & 1.20 & 1.26 & 1.32 & 1.39 & 1.46 & 1.05 \\
1969Q1--2025Q2$^*$ & 0.99 & \textbf{0.98} & 0.99 & 1.01 & 1.05 & 1.09 & 1.14 & 1.20 & 1.26 & 1.34 & 1.02 \\
1969Q1--1999Q4 & \textbf{1.00} & 1.00 & 1.02 & 1.04 & 1.08 & 1.12 & 1.17 & 1.23 & 1.29 & 1.36 & 1.05 \\
2000Q1--2025Q2 & 1.02 & 1.04 & 1.08 & 1.12 & 1.17 & 1.23 & 1.29 & 1.36 & 1.43 & 1.50 & 1.05 \\
2000Q1--2025Q2$^*$ & 0.97 & \textbf{0.96} & 0.96 & 0.97 & 1.00 & 1.04 & 1.09 & 1.16 & 1.23 & 1.31 & 0.97 \\
2000Q1--2008Q4 & 0.99 & \textbf{0.99} & 0.99 & 1.00 & 1.01 & 1.03 & 1.06 & 1.10 & 1.13 & 1.18 & 0.99 \\
2000Q1--2019Q4 & 0.97 & 0.95 & \textbf{0.94} & 0.94 & 0.96 & 0.98 & 1.01 & 1.05 & 1.10 & 1.16 & 0.97 \\
2010Q1--2019Q4 & 0.94 & \textbf{0.93} & 0.99 & 1.10 & 1.24 & 1.41 & 1.59 & 1.79 & 2.00 & 2.21 & 0.94 \\
2022Q1--2025Q2 & \textbf{0.98} & 1.03 & 1.14 & 1.30 & 1.49 & 1.70 & 1.93 & 2.16 & 2.41 & 2.66 & 1.02 \\
\hline
\multicolumn{1}{c}{RGDP}\\
1969Q1--2025Q2 & \textbf{0.99} & 1.00 & 1.01 & 1.03 & 1.06 & 1.10 & 1.15 & 1.21 & 1.26 & 1.33 & 1.01 \\
1969Q1--2025Q2$^*$ & 0.99 & \textbf{0.99} & 0.99 & 1.01 & 1.04 & 1.07 & 1.11 & 1.17 & 1.22 & 1.28 & 1.00 \\
1969Q1--1999Q4 & 0.98 & 0.97 & \textbf{0.97} & 0.98 & 1.00 & 1.02 & 1.06 & 1.10 & 1.15 & 1.20 & 1.01 \\
2000Q1--2025Q2 & \textbf{1.00} & 1.01 & 1.03 & 1.05 & 1.09 & 1.13 & 1.19 & 1.24 & 1.31 & 1.37 & 1.01 \\
2000Q1--2025Q2$^*$ & \textbf{0.99} & 0.99 & 1.00 & 1.02 & 1.05 & 1.09 & 1.14 & 1.19 & 1.25 & 1.32 & 1.00 \\
2000Q1--2008Q4 & 0.98 & \textbf{0.98} & 0.99 & 1.01 & 1.05 & 1.09 & 1.15 & 1.21 & 1.28 & 1.36 & 0.99 \\
2000Q1--2019Q4 & \textbf{1.00} & 1.00 & 1.02 & 1.05 & 1.08 & 1.13 & 1.18 & 1.23 & 1.30 & 1.36 & 1.01 \\
2010Q1--2019Q4 & \textbf{1.00} & 1.01 & 1.03 & 1.06 & 1.09 & 1.14 & 1.19 & 1.25 & 1.32 & 1.39 & 1.01 \\
2022Q1--2025Q2 & 0.98 & \textbf{0.97} & 0.97 & 0.99 & 1.01 & 1.04 & 1.08 & 1.13 & 1.19 & 1.25 & 0.98 \\
\hline
\multicolumn{1}{c}{INDPROD}\\
1969Q1--2025Q2 & 0.98 & \textbf{0.97} & 0.97 & 0.99 & 1.02 & 1.06 & 1.10 & 1.16 & 1.22 & 1.29 & 1.00 \\
1969Q1--2025Q2$^*$ & 0.98 & \textbf{0.97} & 0.97 & 0.99 & 1.02 & 1.06 & 1.11 & 1.16 & 1.23 & 1.30 & 1.00 \\
1969Q1--1999Q4 & 0.98 & \textbf{0.98} & 0.99 & 1.01 & 1.04 & 1.09 & 1.14 & 1.20 & 1.27 & 1.35 & 1.01 \\
2000Q1--2025Q2 & 0.96 & 0.94 & 0.92 & \textbf{0.91} & 0.92 & 0.94 & 0.96 & 1.00 & 1.04 & 1.09 & 0.94 \\
2000Q1--2025Q2$^*$ & 0.96 & 0.94 & 0.92 & \textbf{0.91} & 0.92 & 0.93 & 0.96 & 1.00 & 1.04 & 1.09 & 0.94 \\
2000Q1--2008Q4 & 0.97 & 0.94 & 0.93 & \textbf{0.93} & 0.93 & 0.95 & 0.98 & 1.01 & 1.06 & 1.11 & 0.95 \\
2000Q1--2019Q4 & 0.96 & 0.94 & 0.92 & \textbf{0.91} & 0.92 & 0.94 & 0.96 & 1.00 & 1.04 & 1.09 & 0.94 \\
2010Q1--2019Q4 & 0.94 & 0.89 & \textbf{0.87} & 0.88 & 0.90 & 0.96 & 1.03 & 1.11 & 1.21 & 1.32 & 0.89 \\
2022Q1--2025Q2 & 0.92 & 0.86 & 0.81 & 0.77 & \textbf{0.75} & 0.75 & 0.77 & 0.81 & 0.87 & 0.94 & 0.85 \\
\hline
\multicolumn{1}{c}{CPI}\\
1981Q3--2025Q2 & \textbf{0.99} & 1.00 & 1.02 & 1.04 & 1.08 & 1.12 & 1.18 & 1.24 & 1.30 & 1.37 & 1.00 \\
1981Q3--2025Q2$^*$ & \textbf{0.99} & 0.99 & 1.00 & 1.03 & 1.07 & 1.12 & 1.17 & 1.24 & 1.31 & 1.39 & 1.00 \\
1981Q3--1999Q4 & 0.98 & \textbf{0.97} & 0.98 & 1.01 & 1.05 & 1.10 & 1.17 & 1.25 & 1.33 & 1.42 & 0.99 \\
2000Q1--2025Q2 & \textbf{1.00} & 1.01 & 1.02 & 1.05 & 1.08 & 1.13 & 1.18 & 1.23 & 1.29 & 1.36 & 1.00 \\
2000Q1--2025Q2$^*$ & \textbf{0.99} & 1.00 & 1.01 & 1.04 & 1.07 & 1.12 & 1.18 & 1.24 & 1.31 & 1.38 & 1.00 \\
2000Q1--2008Q4 & 1.01 & 1.02 & 1.03 & 1.05 & 1.07 & 1.09 & 1.12 & 1.15 & 1.18 & 1.22 & 1.01 \\
2000Q1--2019Q4 & 1.01 & 1.03 & 1.05 & 1.09 & 1.13 & 1.18 & 1.24 & 1.30 & 1.37 & 1.44 & 1.01 \\
2010Q1--2019Q4 & \textbf{0.99} & 0.99 & 1.01 & 1.03 & 1.07 & 1.12 & 1.17 & 1.24 & 1.31 & 1.38 & 1.00 \\
2022Q1--2025Q2 & 0.94 & 0.90 & 0.87 & \textbf{0.87} & 0.88 & 0.92 & 0.97 & 1.03 & 1.11 & 1.19 & 0.97 \\
\hline
\end{tabular}
\caption{Relative RMSFE for UNEMP, RGDP, INDPROD, and CPI at forecast horizon $h=4$ (relative to the original mean forecast at the same horizon). Fixed factors, the recursive column, bolding, and starred evaluation samples are defined as in Table~\ref{tab:CorrMeanF}; correction factors are estimated separately for this horizon.}
\label{tab:CorrMeanF-h4}
\end{table}

\clearpage

\begin{figure}[htbp]
    
    \centering
    \includegraphics[width=\textwidth]{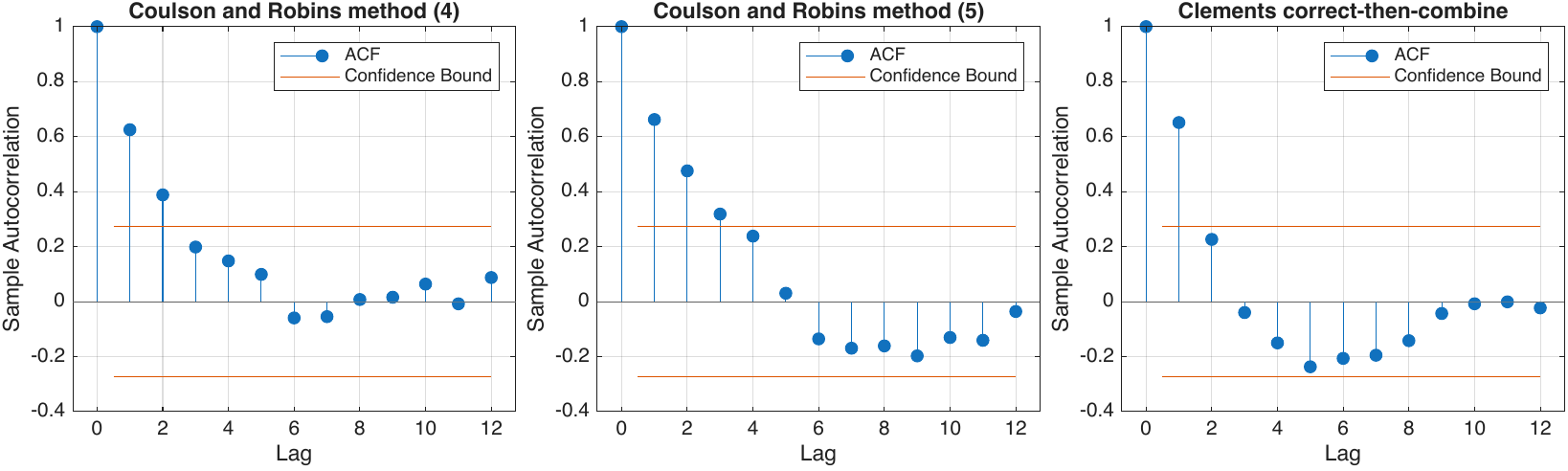}
    \caption{Sample autocorrelation functions of the forecast errors for the additional benchmark methods. From left to right: methods~(4) and~(5) of \citet{CoulsonRobins1993}, and the correct-then-combine approach based on \citet{clements2022}.} 
    \label{fig:autocorr-new-methods}
\end{figure}
}

\end{document}